\documentclass[aps,nopacs,nokeys,11pt,notitlepage,superscriptaddress]{revtex4-2}
\usepackage{amsmath,amsfonts,amssymb}
\usepackage{times}

\usepackage{mathtools}
\usepackage{multirow}
\usepackage{ulem}
\usepackage{algorithmic}
\usepackage[ruled, vlined, linesnumbered]{algorithm2e}
\usepackage[utf8]{inputenc}

\DeclareMathOperator{\Tr}{Tr}
\newcommand{\cH}{\mathcal H}
\newcommand{\cN}{\mathcal N}
\newcommand{\cE}{\mathcal E}
\newcommand{\id}{\mathrm{id}}
\newcommand{\EPO}{\mathrm{EPO}}

\newcommand{\ket}[1]{\lvert #1\rangle}
\newcommand{\ketbra}[2]{\lvert #1\rangle\!\langle #2\rvert}
\newcommand{\vket}[1]{\lvert #1\rangle\!\rangle}
\newcommand{\vbra}[1]{\langle\!\langle #1\rvert}

\usepackage{bm}
\usepackage{times}

\usepackage{mathtools}
\usepackage{amsmath}
\usepackage[shortlabels]{enumitem}
\usepackage{graphicx,epic,eepic,epsfig,amsmath,latexsym,amssymb,verbatim,color}
 
\usepackage{amsfonts}       % blackboard math symbols
\usepackage{nicefrac}       % compact symbols for 1/2, etc.

\usepackage{amsmath}
\usepackage{bbm}

\usepackage{float}
\usepackage{tikz}
\usetikzlibrary{chains}
\usetikzlibrary{fit}
\usetikzlibrary{arrows}
\usepackage{epsfig}
\usetikzlibrary{shapes.symbols,patterns} % for source symbols
\usepackage{pgfplots}
\pgfplotsset{compat=1.18}

\usepackage[strict]{changepage}
\usepackage{hyperref}
\hypersetup{colorlinks=true,citecolor=blue,linkcolor=blue,filecolor=blue,urlcolor=blue,breaklinks=true}

\usepackage[marginal]{footmisc}
\usepackage{url}
\usepackage{theorem}

\newtheorem{definition}{Definition}
\newtheorem{proposition}{Proposition}
\newtheorem{lemma}[proposition]{Lemma}

\newtheorem{theorem}[proposition]{Theorem}

\newtheorem{corollary}[proposition]{Corollary}

\def\squareforqed{\hbox{\rlap{$\sqcap$}$\sqcup$}}
\def\qed{\ifmmode\squareforqed\else{\unskip\nobreak\hfil
\penalty50\hskip1em\null\nobreak\hfil\squareforqed
\parfillskip=0pt\finalhyphendemerits=0\endgraf}\fi}
\def\endenv{\ifmmode\;\else{\unskip\nobreak\hfil
\penalty50\hskip1em\null\nobreak\hfil\;
\parfillskip=0pt\finalhyphendemerits=0\endgraf}\fi}
\newenvironment{proof}{\noindent \textbf{{Proof~} }}{\hfill $\blacksquare$}

\newcounter{remark}

\newcounter{example}

\mathchardef\ordinarycolon\mathcode`\:
\mathcode`\:=\string"8000
\def\vcentcolon{\mathrel{\mathop\ordinarycolon}}
\begingroup \catcode`\:=\active
  \lowercase{\endgroup
  \let :\vcentcolon
  }

\usepackage{cleveref}
\usepackage{graphicx}
\usepackage{xcolor}

\RequirePackage[framemethod=default]{mdframed}
\newmdenv[skipabove=7pt,
skipbelow=7pt,
backgroundcolor=darkblue!15,
innerleftmargin=5pt,
innerrightmargin=5pt,
innertopmargin=5pt,
leftmargin=0cm,
rightmargin=0cm,
innerbottommargin=5pt,
linewidth=1pt]{tBox}

\newmdenv[skipabove=7pt,
skipbelow=7pt,
backgroundcolor=blue2!25,
innerleftmargin=5pt,
innerrightmargin=5pt,
innertopmargin=5pt,
leftmargin=0cm,
rightmargin=0cm,
innerbottommargin=5pt,
linewidth=1pt]{dBox}
\newmdenv[skipabove=7pt,
skipbelow=7pt,
backgroundcolor=darkkblue!15,
innerleftmargin=5pt,
innerrightmargin=5pt,
innertopmargin=5pt,
leftmargin=0cm,
rightmargin=0cm,
innerbottommargin=5pt,
linewidth=1pt]{sBox}
\definecolor{darkblue}{RGB}{0,76,156}
\definecolor{darkkblue}{RGB}{0,0,153}
\definecolor{blue2}{RGB}{102,178,255}
\definecolor{darkred}{RGB}{195,0,0}
\newcommand{\nc}{\newcommand}
\nc{\rnc}{\renewcommand}
\nc{\lbar}[1]{\overline{#1}}
\nc{\proj}[1]{| #1\rangle\!\langle #1 |}
\nc{\avg}[1]{\langle#1\rangle}
\nc{\rank}{\operatorname{Rank}}
\nc{\smfrac}[2]{\mbox{$\frac{#1}{#2}$}}
\nc{\tr}{\operatorname{Tr}}
\nc{\ox}{\otimes}
\nc{\dg}{\dagger}
\nc{\dn}{\downarrow}
\nc{\cA}{{\cal A}}
\nc{\cB}{{\cal B}}
\nc{\cC}{{\cal C}}
\nc{\cD}{{\cal D}}
\nc{\cF}{{\cal F}}
\nc{\cG}{{\cal G}}
\nc{\cI}{{\cal I}}
\nc{\cJ}{{\cal J}}
\nc{\cK}{{\cal K}}
\nc{\cL}{{\cal L}}
\nc{\cM}{{\cal M}}
\nc{\cO}{{\cal O}}
\nc{\cP}{{\cal P}}
\nc{\cQ}{{\cal Q}}
\nc{\cR}{{\cal R}}
\nc{\cS}{{\cal S}}
\nc{\cT}{{\cal T}}
\nc{\cU}{{\cal U}}
\nc{\cV}{{\cal V}}
\nc{\cX}{{\cal X}}
\nc{\cY}{{\cal Y}}
\nc{\cZ}{{\cal Z}}
\nc{\cW}{{\cal W}}
\nc{\csupp}{{\operatorname{csupp}}}
\nc{\qsupp}{{\operatorname{qsupp}}}
\nc{\var}{{\operatorname{var}}}
\nc{\rar}{\rightarrow}
\nc{\lrar}{\longrightarrow}
\nc{\polylog}{{\operatorname{polylog}}}
\nc{\wt}{{\operatorname{wt}}}
\nc{\supp}{{\operatorname{supp}}}

\nc{\argmin}{{\operatorname{argmin}}}

\def\x{\xi}

\nc{\RR}{{{\mathbb R}}}
\nc{\CC}{{{\mathbb C}}}
\nc{\FF}{{{\mathbb F}}}
\nc{\NN}{{{\mathbb N}}}
\nc{\ZZ}{{{\mathbb Z}}}
\nc{\PP}{{{\mathbb P}}}
\nc{\QQ}{{{\mathbb Q}}}
\nc{\UU}{{{\mathbb U}}}
\nc{\EE}{{{\mathbb E}}}
\nc{\CHSH}{{\operatorname{CHSH}}}

\nc{\<}{\langle}
\rnc{\>}{\rangle}
\nc{\rU}{\mbox{U}}

\nc{\ob}[1]{#1}

\nc{\SEP}{{\text{\rm SEP}}}
\nc{\NS}{{\text{\rm NS}}}
\nc{\LOCC}{{\text{\rm LOCC}}}
\nc{\PPT}{{\text{\rm PPT}}}
\nc{\EXT}{{\text{\rm EXT}}}
\nc{\Sym}{{\operatorname{Sym}}}

\nc{\ERLO}{{E_{\text{r,LO}}}}
\nc{\ERLOCC}{{E_{\text{r,LOCC}}}}
\nc{\ERPPT}{{E_{\text{r,PPT}}}}
\nc{\ERLOCCinfty}{{E^{\infty}_{\text{r,LOCC}}}}
\nc{\Aram}{{\operatorname{\sf A}}}

\usepackage{tikz}
\usepackage{hyperref}
\hypersetup{colorlinks=true,citecolor=blue,linkcolor=blue,filecolor=blue,urlcolor=blue,breaklinks=true}

\makeatletter
\def\grd@save@target#1{%
  \def\grd@target{#1}}
\def\grd@save@start#1{%
  \def\grd@start{#1}}
\tikzset{
  grid with coordinates/.style={
    to path={%
      \pgfextra{%
        \edef\grd@@target{(\tikztotarget)}%
        \tikz@scan@one@point\grd@save@target\grd@@target\relax
        \edef\grd@@start{(\tikztostart)}%
        \tikz@scan@one@point\grd@save@start\grd@@start\relax
        \draw[minor help lines,magenta] (\tikztostart) grid (\tikztotarget);
        \draw[major help lines] (\tikztostart) grid (\tikztotarget);
        \grd@start
        \pgfmathsetmacro{\grd@xa}{\the\pgf@x/1cm}
        \pgfmathsetmacro{\grd@ya}{\the\pgf@y/1cm}
        \grd@target
        \pgfmathsetmacro{\grd@xb}{\the\pgf@x/1cm}
        \pgfmathsetmacro{\grd@yb}{\the\pgf@y/1cm}
        \pgfmathsetmacro{\grd@xc}{\grd@xa + \pgfkeysvalueof{/tikz/grid with coordinates/major step}}
        \pgfmathsetmacro{\grd@yc}{\grd@ya + \pgfkeysvalueof{/tikz/grid with coordinates/major step}}
        \foreach \x in {\grd@xa,\grd@xc,...,\grd@xb}
        \node[anchor=north] at (\x,\grd@ya) {\pgfmathprintnumber{\x}};
        \foreach \y in {\grd@ya,\grd@yc,...,\grd@yb}
        \node[anchor=east] at (\grd@xa,\y) {\pgfmathprintnumber{\y}};
      }
    }
  },
  minor help lines/.style={
    help lines,
    step=\pgfkeysvalueof{/tikz/grid with coordinates/minor step}
  },
  major help lines/.style={
    help lines,
    line width=\pgfkeysvalueof{/tikz/grid with coordinates/major line width},
    step=\pgfkeysvalueof{/tikz/grid with coordinates/major step}
  },
  grid with coordinates/.cd,
  minor step/.initial=.2,
  major step/.initial=1,
  major line width/.initial=2pt,
}
\makeatother

\usepackage{thmtools}
\usepackage{thm-restate}
\usepackage{etoolbox}
\makeatletter
\def\problem@s{}
\newcounter{problems@cnt}

\newcommand{\allproblems}{\problem@s}
\makeatother

\definecolor{beamer}{rgb}{0.2,0.2,0.7}
\definecolor{colorone}{rgb}{1,0.36,0.03}
\definecolor{colortwo}{rgb}{0.4,0.77,0.17}
\definecolor{colorthree}{rgb}{0.01,0.51,0.93}
\definecolor{colorfour}{rgb}{0.47,0.26,0.58}
\definecolor{colorfive}{rgb}{0.12,0.55,0.16}
\usepackage{tcolorbox}
\usepackage{relsize}
\usepackage{graphicx}
\usepackage{booktabs}
\usepackage{amsmath}
\usepackage{tikz}
\usetikzlibrary{quantikz}
\nc{\st}{\text{subject to} \ }
\nc{\supre}{\text{supremum} \ }
\nc{\sdp}{\text{sdp}}
\usepackage{array}

\allowdisplaybreaks
\nc{\ith}[1]{{#1}^\mathrm{th}}
\begin{document}

\title{Universal quantum state purification under energy-preserving constraints}

\author{Xing-Chen Guo}
\affiliation{Thrust of Artificial Intelligence, Information Hub, The Hong Kong University of Science and Technology (Guangzhou), Guangzhou 511453, China}
\affiliation{School of Mathematics, South China University of Technology, Guangzhou 510641, China}

\author{Benchi Zhao}
\email{benchizhao@gmail.com}
\affiliation{QICI Quantum Information and Computation Initiative, Department of Computer Science, The University of Hong Kong, Pokfulam Road, Hong Kong}

\author{Xin Wang}
\email{felixxinwang@hkust-gz.edu.cn}
\affiliation{Thrust of Artificial Intelligence, Information Hub, The Hong Kong University of Science and Technology (Guangzhou), Guangzhou 511453, China}

\begin{abstract}
Pure quantum states are a fundamental resource for quantum information processing, but unavoidable noise degrades their purity and limits the performance of quantum tasks. Quantum state purification attempts to restore purity by jointly processing several noisy copies of an unknown state, typically in a probabilistic way. Whether purification is possible, and how much it can achieve, depends strongly on the set of operations that are physically allowed. In this work, we study universal purification under \textit{energy-preserving operations}, which cannot exchange energy with an external environment. We identify exactly when the energy-preserving constraint forbids any universal purification advantage, and we provide concrete examples in which an advantage available to unrestricted operations disappears under energy conservation. When energy-preserving purification remains possible, we determine its optimal fidelity and success probability and construct the corresponding protocols. The framework includes unrestricted purification as a special case and extends naturally to purification assisted by an external energy resource.

% Quantum state purification probabilistically converts multiple noisy copies into an output of higher fidelity, thereby restoring state purity from quantum noises. Whether such a purification is possible and how effective it is, however, depends on which operations are physically allowed. Here we study universal purification under \textit{energy-preserving operations}, isolating the purification achievable through collective information processing alone without drawing on an external energy supply. For a class of noise channels including depolarizing ones, we give a necessary and sufficient condition for determining when the energy-preserving restriction forbids universal purification. Concrete examples justify that the energy-preserving constraint alone can indeed eliminate a purification advantage available to unrestricted operations, providing a genuine no-go result for universal purification under specific energy-preserving constraints. For the case where energy-preserving universal purification is possible, we identify the optimal purification performance and the associated purification protocol. We further construct explicit energy-preserving implementations of the optimal protocols from their Choi operators and illustrate their performance numerically. Our framework includes unrestricted purification as a special case and extends seamlessly to settings assisted by a specified external energy resource. These results identify the fundamental purification limits imposed by energy conservation and characterize the optimal energy-preserving purifiers in principle.
\end{abstract}
% \date{\today}
\maketitle

%%%%%%%%%%%%%%%%%%%%%%%%%%%%%%%%%%%%%%%%%%%%%%%%%%%%%%%%%%%%
\onecolumngrid
\section{Introduction}

Pure states are essential to quantum information processing, but unavoidable noise degrades their purity and limits the performance of quantum information tasks~\cite{divincenzo2000physical,preskill2018quantum,guo2023noise,rower2025qubit}. Increasing state purity is therefore a central challenge for quantum technologies. Alongside quantum error correction, quantum state purification offers another route: it processes multiple noisy copies to produce an output of higher fidelity. Because useful gains can already be obtained from modest numbers of qudits, purification is particularly relevant to resource-limited settings~\cite{bennett1996purification,cirac1999optimal,keyl2001rate,childs2025streaming}.

Purification appears in two related operational settings. When the desired resource is specified, distillation protocols exploit its known structure to extract higher-quality entangled states or magic states from noisy inputs~\cite{bennett1996purification,Devetak2003a,Fang2019c,zhao2021practical,Briegel1998,Dur1999,Azuma2015a,bravyi2012magic,itogawa2025efficient,sales2025experimental}. Universal purification addresses a different task: the underlying pure state is unknown, and one seeks a single protocol that increases the average output fidelity over arbitrary pure inputs affected by a given noise channel~\cite{Barenco1997,cirac1999optimal,keyl2001rate,fiuravsek2004optimal,childs2025streaming,yao2025protocols,liu2025no,Li2024a}. In the probabilistic setting, output fidelity and success probability quantify complementary aspects of performance, namely quality and yield~\cite{fiuravsek2004optimal,yao2025protocols}.

Purifiability, however, is not a property of the noisy ensemble alone. A purification protocol is a physical transformation, and its attainable performance depends on the available class of operations. General no-go results already show that restrictions on state transformations can fundamentally limit purification, even when postselection is allowed~\cite{fang2020no}. In universal purification, restricting the protocol to local operations and classical communication or to classically simulable operations likewise changes which fidelity improvements are accessible~\cite{zhao2025power,he2025no}. Thus, before asking how much purity can be recovered, one must specify which operations are physically permitted.

Energy conservation provides a fundamental and operationally well-defined restriction. At microscopic scales, energy conservation restricts the possible state transitions, while an explicit work-storage system accounts for the energy supplied to or extracted from the system~\cite{horodecki2013fundamental}. Here we restrict the purifier to energy-preserving operations. They preserve the complete measurement statistics of a fixed Hamiltonian, and can be implemented without exchanging energy with an external system~\cite{chiribella2017optimal}. This restriction isolates the purification achievable through collective information processing alone. It also makes the central questions precise: when do energy-preserving constraints forbid universal purification, and what performance remains attainable when such a purification is possible?

We answer these questions for a class of noise channels that includes depolarizing noises. We derive a necessary and sufficient condition for the impossibility of universal energy-preserving purification and give explicit examples in which the energy constraint eliminates a purification advantage available to unrestricted operations. When purification is possible, we determine the maximum average fidelity and the maximum average success probability at that fidelity, identify the corresponding protocols, and construct energy-preserving implementations from their Choi operators. Numerical examples illustrate the resulting performance. The framework recovers unrestricted purification as a special case and extends to purification assisted by a specified external energy resource. These results establish the fundamental purification limits imposed by energy conservation and separate the gains obtainable from collective information processing from those that require energetic assistance.
%%%%%%%%%%%%%%%%%%%%%%%%%%%%%%%%%%%%%%%%%%%%%%%%%%%%%%%%%%

\section{Energy-preserving operations}

We begin with the definition of energy-preserving operations. A detailed study of which can be found in Ref.~\cite{chiribella2017optimal}. For convenience, we adopt some new terminology and notation in this paper.   

\begin{definition}[Energy-preserving operations]
    For Hamiltonian $H \in \mathcal{B}(\mathcal{H})$, a channel $\Lambda \in \text{CPTP}(\mathcal{B}(\mathcal{H}), \mathcal{B}(\mathcal{H}))$ is termed energy-preserving for $H$ if
    \begin{equation}
        \Lambda^\dagger(H^n) = H^n, \quad\quad \forall n \in \mathbb{N}.
    \end{equation}
    An instrument $\{\Lambda_x\}_{x}$, with $\Lambda_x \in \text{CPTN}(\mathcal{B}(\mathcal{H}), \mathcal{B}(\mathcal{H}))$ and $\sum_x \Lambda_x \in \text{CPTP}(\mathcal{B}(\mathcal{H}), \mathcal{B}(\mathcal{H}))$, is energy-preserving for $H$ if $\sum_x \Lambda_x$ is an energy-preserving channel for $H$. An operation $\Lambda \in \text{CPTN}(\mathcal{B}(\mathcal{H}), \mathcal{B}(\mathcal{H}))$ is energy-preserving for $H$ if it belongs to some energy-preserving instrument for $H$. We denote the set of all such operations by $\text{EPO}(H)$.
\end{definition}

Physically, a channel $\Lambda$ is energy-preserving for $H$, if and only if when acting on a system with Hamiltonian $H$, it leaves the measurement statistics of the energy observable unchanged, which captures the meaning of \textit{energy-preserving}. 

Note that energy-preserving operations can indeed be implemented in an energy-preserving manner: system-environment interactions followed by projective measurements that all preserve the energy of the system. This fact is formalized as follows.

\begin{lemma}[Ref.~\cite{chiribella2017optimal}]
\label{lem1}
Let $\Lambda \in \text{CPTN}(\mathcal{B}(\mathcal{H}), \mathcal{B}(\mathcal{H}))$ and let $H \in \mathcal{B}(\mathcal{H})$ be a Hamiltonian with spectral decomposition $H = \sum_E E P_E$. Then $\Lambda \in \text{EPO}(H)$ if and only if one of the following conditions holds:
\begin{enumerate}[leftmargin=*, align=left]
\item \textbf{Stinespring:} $\Lambda$ admits an implementation of the form
    \[
    \Lambda(\rho) = \operatorname{Tr}_{\mathcal{H}_{\text{env}}} \left[ I \otimes Q_{\text{env}} \cdot\, U (\rho \otimes \ket{\phi_1}\bra{\phi_1}) U^\dagger \right], \ \ \forall \rho \in \mathcal{B}(\mathcal{H}),
    \]
    where:
    \begin{itemize}
        \item $\ket{\phi_1}$ is a ground state of an environment's Hamiltonian $H_{\text{env}} \in \mathcal{B}(\mathcal{H}_{\text{env}})$,
        \item $U$ is a unitary operation satisfying $[U, H \otimes I_{\mathcal{H}_{\text{env}}}] = 0$ and $[U, I \otimes H_{\text{env}}] = 0$,
        \item $Q_{\text{env}}$ is a projector on $\mathcal{H}_{\text{env}}$ with $[Q_{\text{env}}, H_{\text{env}}] = 0$.
    \end{itemize}

\item \textbf{Kraus:} There exists a set of Kraus operators $\{M_k\}_k$ for $\Lambda$ such that $[M_k, H] = 0$ for all $k$.

\item \textbf{Choi:} The Choi operator $\Gamma_\Lambda$ of $\Lambda$ satisfies $\Gamma_\Lambda=\Pi\Gamma_\Lambda\Pi$, where $\Pi := \sum_E P_E^T\otimes P_E$.
\end{enumerate}
\end{lemma}

%%%%%%%%%%%%%%%%%%%%%%%%%%%%%%%%%%%%%%%%%%%%%%%
\section{Universal energy-preserving purification protocol}
% We start from the definition of energy-preserving operation. For a given Hamiltonian $H$, a channel $\Lambda \in \text{CPTP}(\mathcal{B}(\mathcal{H}), \mathcal{B}(\mathcal{H}))$ is termed energy-preserving for $H$ if
% $\Lambda^\dagger(H^n) = H^n, \,\forall n \in \mathbb{N}$, 
% where $\mathcal{B}(\mathcal{H})$ stands for the linear operator space of the Hilbert space $\cH$.
% An instrument $\{\Lambda_x\}_{x}$, with $\Lambda_x \in \text{CPTN}$ and $\sum_x \Lambda_x \in \text{CPTP}$, is energy-preserving for $H$ if $\sum_x \Lambda_x$ is an energy-preserving channel for $H$. An operation $\Lambda \in \text{CPTN}$ is energy-preserving for $H$ if it belongs to some energy-preserving instrument for $H$~\cite{chiribella2017optimal}. Here we denote the set of all such operations by $\text{EPO}(H)$.

Let $\cE = \tr_{k+1,\cdots,n}\circ\Lambda$ be an $n\rightarrow k$ CPTN map, where the energy-preserving operation $\Lambda\in\mathrm{EPO}(H)$ acts on the $n$-copy system before the last $n-k$ subsystems are discarded. We impose the energy-preserving constraint on this pre-discarding operation $\Lambda$ and focus on $k=1$. Given $n$ copies of $\cN(\psi)$, where $\psi$ is a pure state and $\cN$ is a quantum noise channel, the output state after applying $\cE$ is
\begin{equation}
    \sigma_\psi := \frac{\hat{\sigma}_\psi}{p_\psi} = \frac{\cE(\cN(\psi)^{\ox n})}{\tr[\cE(\cN(\psi)^{\ox n})]},
\end{equation}
where $\hat{\sigma}_\psi=\cE(\cN(\psi)^{\ox n})$ stands for the unnormalized purified state, and $p_\psi=\tr[\cE(\cN(\psi)^{\ox n})]$ refers to the success probability of getting the state $\sigma_\psi$. If the output state $\sigma_\psi$ is closer to the target state $\psi$, i.e., $F(\sigma_\psi, \psi) >  F(\cN(\psi), \psi)$ where $F(\rho,\sigma)$ refers to the fidelity between two states $\rho$ and $\sigma$, then the map $\cE$ is an $n\rightarrow 1$ \textit{energy-preserving purification protocol} for quantum noise $\cN$ and state $\psi$ over Hamiltonian $H$. The diagram of the framework is shown in Fig.~\ref{fig:diagram}. If the map $\cE$ could increase the fidelity on average for all pure states, we would call it an $n\rightarrow 1$ \textit{universal energy-preserving purification protocol}. Specifically, we have the following Definition~\ref{def1}.

\begin{figure}[t]
    \centering 
    \includegraphics[width=0.45\textwidth]{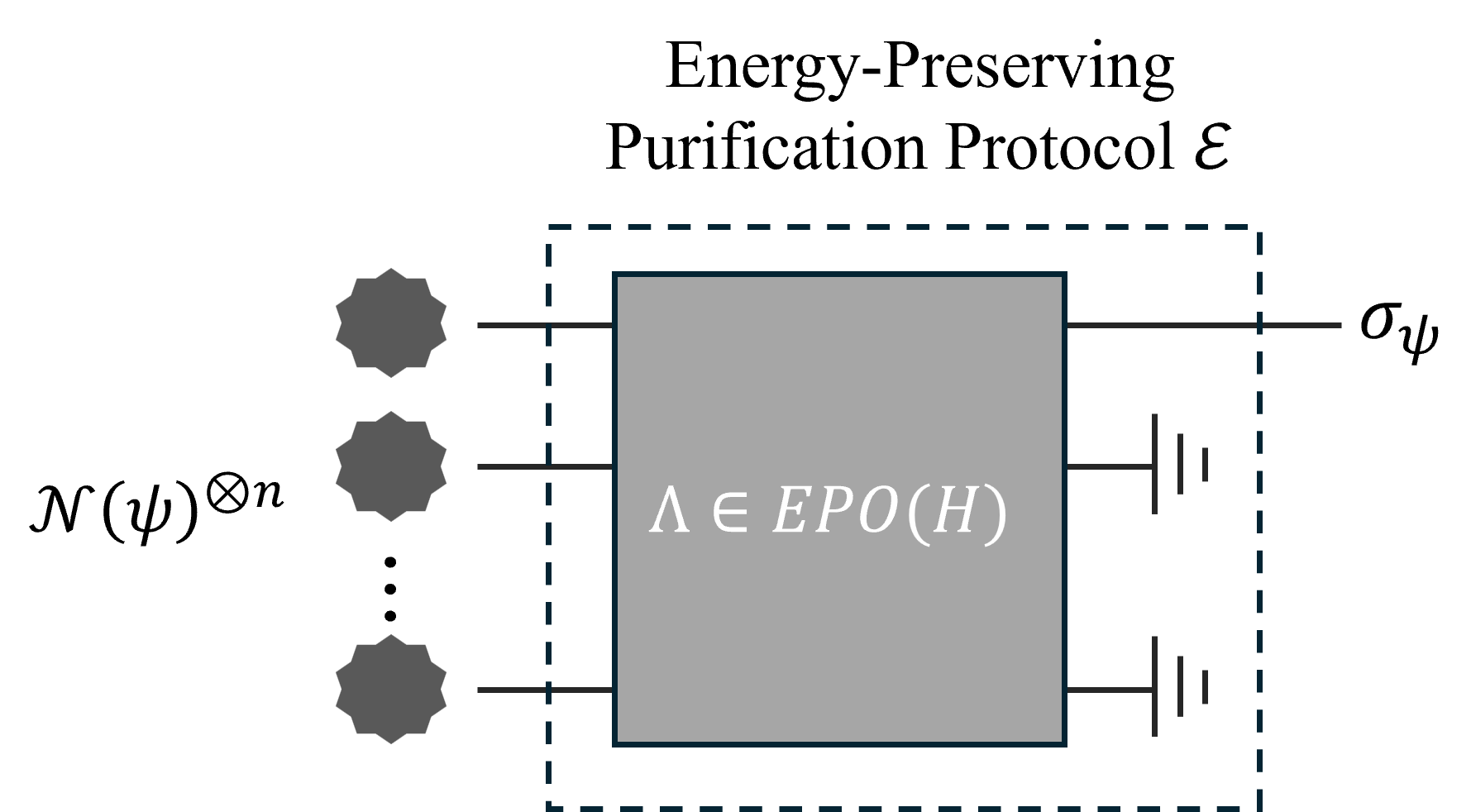}
    \caption{The framework of $n\rightarrow1$ energy-preserving purification. For a given Hamiltonian $H$, the purification protocol $\cE$ applies an energy-preserving operation with respect to the Hamiltonian $\Lambda\in \text{EPO}(H)$ and outputs the first system $\sigma_\psi$, whose fidelity is higher, and partial traces the rest of the systems.}
    \label{fig:diagram} 
\end{figure}

\begin{definition}[Universal energy-preserving purification]
\label{def1}
Let $\mathcal{N} \in \text{CPTP}(\mathcal{B}(\mathcal{H}), \mathcal{B}(\mathcal{H}))$ be the noise channel and $H \in \mathcal{B}(\mathcal{H}^{\otimes n})$ be the Hamiltonian of the $n$-copy system. An $n\rightarrow 1$ map $\mathcal{E} \in \text{CPTN}(\mathcal{B}(\mathcal{H}^{\otimes n}), \mathcal{B}(\mathcal{H}))$ is a universal energy-preserving purification protocol for the triplet $(\mathcal{N}, H, n)$, if all the following conditions are satisfied:
\noindent\begin{enumerate}
    \item \textbf{Zero energy cost:} The protocol is realizable without consuming energy resources, meaning it admits a decomposition $\mathcal{E} = \operatorname{Tr}_{2\cdots n} \circ \Lambda$ for some $\Lambda \in \text{EPO}(H)$;
    \item \textbf{Positive average success probability:} The protocol succeeds with a strictly positive probability on average,
    \begin{equation}
        \Bar{p}(\cN,\cE,n) \coloneqq \int \text{d}\psi \, \operatorname{Tr}[\mathcal{E}(\mathcal{N}(\psi)^{\otimes n})] > 0;
    \end{equation}
    \item \textbf{Increase average fidelity:} The protocol yields an average fidelity strictly greater than the raw fidelity $F_{\rm raw}(\cN)$ of the unpurified noisy states,
    \begin{align}\label{eq:ave_fid}
    \Bar{F}(\cN,\cE,n) &\coloneqq \frac{\int \text{d}\psi \, \operatorname{Tr}\!\left[\psi\,\cE(\cN(\psi)^{\otimes n})\right]}{\Bar{p}(\cN,\cE,n)} > F_{\rm raw}(\cN),\\
    F_{\rm raw}(\cN)&\coloneqq\int d\psi \,\operatorname{Tr}[\psi\cN(\psi)],
    % &=\frac{\int \text{d}\psi \, F(\operatorname{Tr}_{2\cdots n} \circ id(\cN(\psi)^{\otimes n}), \psi)}{\Bar{p}(\cN,\operatorname{Tr}_{2\cdots n} \circ id, H,n)},
    \end{align}
\end{enumerate}
where all integrals are over pure states with respect to the normalized Haar measure. And we call $\Bar{p}(\cN,\cE,n)$ the average success probability and $\Bar{F}(\cN,\cE,n)$ the average fidelity. The set of all universal energy-preserving purification protocols is termed \textit{UEPP}$(\cN, H, n)$.
\end{definition}

%\begin{definition}\label{def:universal purification} \text{(Universal energy-preserving purification protocol)}
%Let $\mathcal{N} \in \text{CPTP}(\mathcal{B}(\mathcal{H}), \mathcal{B}(\mathcal{H}))$ be a non-trivial depolarizing noise and $H\in\mathcal{B}(\mathcal{H}^{\otimes n})$ an n-qudit Hamiltonian with spectral decomposition $H=\sum_EEP_E$. Then an $n\rightarrow 1$ CPTN map $\cE=\tr_{2,\cdots,n}\circ \Lambda$ is a universal energy-preserving purification protocol for triple $(\mathcal{N}, H, n)$, if $\Lambda$ is energy-preserving respect to $H$ and increases the average fidelity for all pure states, i.e., 
%\begin{align}\label{eq:ave_fid}
    %\Bar{F}(\cN,\cE, H,n) &= \frac{\int \text{d}\psi \,\operatorname{Tr}[\psi\cdot\cE(\mathcal{N}(\psi)^{\otimes n})]}{\Bar{p}(\cN,\cE, H,n)} \\
    %&\ge \int d\psi \, F(\cN(\psi), \psi),
%\end{align}
%where $\Bar{p}(\cN,\cE, H,n)=\int \text{d}\psi \,\operatorname{Tr}[\cE(\mathcal{N}(\psi)^{\otimes n})]$ is the average success probability. The integration is taken over all pure states under the Haar measure.
%\end{definition}

%In Definition~\ref{def:universal purification}, the average fidelity is calculated by the counting in the successful cases only. Specifically, quantum pure states uniformly were sampled, apply the purification protocol was applied, all failure cases were ruled out, and the average fidelity of successful cases could be computed by Eq.~\ref{eq:ave_fid}. 

Definition~\ref{def1} establishes the physical criteria for a valid purification machine operating under strict energy conservation. The requirement of zero energy cost imposes thermodynamic neutrality, implying that the restoration of purity arises solely from the collective information processing of the noisy copies rather than through the injection of external energy. Furthermore, the stipulations regarding success probability and fidelity ensure that the protocol is operationally viable and provides a genuine advantage over raw noisy states, distinguishing effective purification from trivial operations that fail to mitigate quantum noise.

The results below apply to every noise channel satisfying
\begin{equation}\label{eq:noise-condition}
 \int d\psi\,\cN(\psi)^{\otimes n}=\cN^{\otimes n}(s_n)>0,
\end{equation}
where $s_n=\int d\psi\,\psi^{\otimes n}$ is the maximally mixed state on the symmetric subspace. Equation~\eqref{eq:noise-condition} requires the Haar-averaged noisy $n$-copy ensemble to have full support. The nontrivial depolarizing channel $\cN_\gamma(\psi)=(1-\gamma)I/d+\gamma\psi$, $0<\gamma<1$, is an important special case and is retained in the analytic and numerical examples.

\begin{theorem}
\label{thm1}
Let $\cN$ satisfy Eq.~\eqref{eq:noise-condition}, and let $H$ be a given Hamiltonian. There exists no $n\rightarrow1$ universal energy-preserving purification protocol
% for the triple $(\mathcal{N},H,n)$
if and only if
\begin{equation}\label{thm1-1}
C^{-\frac{1}{2}}P_m C^{\frac{1}{2}}\cdot \Gamma_{id}\cdot C^{\frac{1}{2}} P_m C^{-\frac{1}{2}} = \Gamma_{id},
\end{equation}
where $\Gamma_{id}$ is the Choi operator of the identity channel $id$, 
% Encapsulating the information of $(\mathcal{N}, H, n)$, 
and the operators $A$ and $C$ are defined as
\begin{align}
    A &= (\mathcal{N}^{\otimes n}\otimes id\ (s_{n+1}))^{\text{T}_{1\cdots n}} \otimes I_{\mathcal{H}^{\otimes (n-1)}}, \\ 
    C &= \Pi\cdot (\mathcal{N}^{\otimes n}(s_n))^{\text{T}} \otimes I_{\mathcal{H}^{\otimes n}}\cdot \Pi,
\end{align}
with $\Pi = \sum_E P_E^T \otimes P_E$ constructed from the spectral decomposition $H=\sum_E E P_E$; $s_k$ denotes the maximally mixed state on the symmetric subspace of $\mathcal{H}^{\otimes k}$; $\text{T}_{1\cdots n}$ is the partial transpose over the first $n$ subsystems; $P_m$ is the projector on the eigenspace of $C^{-\frac{1}{2}} A C^{-\frac{1}{2}}$ with maximum eigenvalue.
\end{theorem}

%In this work, we focus on a specific class of noises $\mathcal{N}$ that satisfy the condition
%\begin{equation}
%\label{eq:noise condition}
%\int \text{d}\psi \, \mathcal{N}(\psi)^{\otimes n}= \mathcal{N}^{\otimes n}(s_n)> 0,
%\end{equation}
%\bc{@xingchen guo hide it into proof, and in the main text, just say we study depo noise}
%for some positive integer $n$, where $s_n$ denotes the maximally mixed state on the symmetric subspace. Such a noise is very common on quantum devices, and the well-known depolarizing noise is one example, which
%is a canonical quantum noise process defined as $\cN_q(\psi)= q\psi +(1-q) \frac{I}{d}$, where $I$ is the identity matrix, and $d$ is the Hilbert space dimension. Besides, the noise in the form of Eq.~\ref{eq:noise condition} has a good property of symmetry, which is easy to analyze.

The proof of Theorem~\ref{thm1} begins by reformulating the average fidelity of any operation as a ratio of traces involving its Choi operator and two structural matrices that encode the noise and energy constraint, allowing for spectral analysis. For the \textit{only if} direction, we show that if no energy-preserving protocol outperforms the trivial strategy of doing nothing, the identity channel must itself produce the maximum average fidelity, forcing its Choi operator to lie within a specific eigenspace defined by the structural matrices. Conversely, for the \textit{if} direction, we demonstrate that satisfying this eigenspace condition implies that the identity channel already reaches maximum average fidelity, thereby mathematically precluding the existence of any superior purification protocol. The derivation is detailed in Appendix~\ref{supp_sec:B}.

Theorem~\ref{thm1} characterizes exactly when the energy-preserving restriction precludes universal purification. When its criterion holds, no energy-preserving protocol can increase the average fidelity beyond the raw fidelity, making the identity strategy optimal. Operationally, the theorem provides a screening test, determined by the noise and Hamiltonian, for regimes in which no search for a nontrivial energy-preserving purifier can succeed.

The remaining question is whether this no-go criterion is realized by concrete Hamiltonians. The following proposition answers this question affirmatively for a family of nondegenerate 2-qubit Hamiltonians. Since the same noisy 2-copy ensemble admits a strict fidelity improvement under unrestricted operations, the result establishes a genuine no-go imposed specifically by the energy-preserving constraint.

\begin{proposition}[No-go Hamiltonian family]
\label{prop:n2-family-concise}
Let $\cH=\mathbb C^2$, $n=2$, and $\cN_\gamma(\psi)=\gamma\psi+(1-\gamma)\frac{I}{2}$ be the depolarizing qubit noise with parameter $0<\gamma<1$. Consider the 2-qubit Hamiltonian
\begin{equation}
 H=aI+bX\otimes I+cI\otimes Z+dX\otimes Z,
 \qquad a,b,c,d\in\mathbb R,
 \label{eq:n2c-H}
\end{equation}
where $\lvert{b}\rvert, \lvert{c}\rvert, \lvert{d}\rvert$ are 3 distinct numbers. Then the necessary and sufficient no-go condition Eq.~\eqref{thm1-1} is satisfied, and hence no
$2\rightarrow1$ universal energy-preserving purification protocol exists in this case. However, without the energy-preserving constraint, the optimal $2\rightarrow1$ universal purification protocol produces an average fidelity
\begin{equation}
 F_{\max}^{\mathrm{all}}(\cN_\gamma,2)
 =\frac{(1+\gamma)(3+\gamma)}{2(3+\gamma^2)},
 \label{eq:n2c-unrestricted}
\end{equation}
which is strictly greater than the raw fidelity $F_{\text{raw}}(\mathcal{N}_\gamma)=(1+\gamma)/2$. Thus universal purification exists for the same noisy 2-copy ensemble if the energy constraint is removed.
\end{proposition}

\begin{proof}
For $s,t\in\{+1,-1\}$, let $\ket{s_x}$ and $\ket{t_z}$ be the eigenvectors
of Pauli operators $X$ and $Z$ with eigenvalues $s$ and $t$, respectively, and define
$\ket{e_{s,t}}:=\ket{s_x}\otimes\ket{t_z}$.  The four vectors
$\ket{e_{s,t}}$ are eigenvectors of $H$ with energies
$a+bs+ct+dst$. Since $\lvert{b}\rvert, \lvert{c}\rvert, \lvert{d}\rvert$ are 3 distinct numbers, $H$ must be nondegenerate, and we have
\begin{equation}
 \Pi=\sum_{s,t=\pm1}
\ketbra{e_{s,t}^*}{e_{s,t}^*}\otimes\ketbra{e_{s,t}}{e_{s,t}}.
 \label{eq:n2c-Pi}
\end{equation}
Note that $\Pi$ has no dependence on the concrete values of the parameters $a,b,c,d$. Then it is routine to check that
\begin{equation}
 C=\Pi\left[\cN_\gamma^{\otimes2}(s_2)^{T}
 \otimes I_{\cH^{\otimes2}}\right]\Pi=\frac14\Pi.
 \label{eq:n2c-C}
\end{equation}
Denote $\vket{E_{s,t}}:=\ket{e_{s,t}}\otimes\ket{e_{s,t}}$ and $\vket{\phi_t}:=\frac{\vket{E_{+,t}}+\vket{E_{-,t}}}{\sqrt2}$, we have
\begin{equation}
 P_m=\sum_{t=\pm1}\vket{\phi_t}\vbra{\phi_t}.
 \label{eq:n2c-Pm}
\end{equation}
and
\begin{equation}
 \qquad \Gamma_{\id}=(\sum_{s,t=\pm1}\vket{E_{s,t}})\cdot(\sum_{s,t=\pm1}\vbra{E_{s,t}}).
 \label{eq:n2c-id}
\end{equation}
Then $P_m\Gamma_{\id}P_m=\Gamma_{\id}$ holds.  Since
$C^{1/2}=\Pi/2$ and $C^{-1/2}=2\Pi$, we have
\begin{equation}
 C^{-1/2}P_mC^{1/2}\Gamma_{\id}
 C^{1/2}P_mC^{-1/2}=\Pi P_m\Pi\cdot\Gamma_{\id}\cdot
 \Pi P_m\Pi=P_m\Gamma_{\id}P_m=\Gamma_{\id}.
 \label{eq:n2c-condition5}
\end{equation}
Thus the necessary and sufficient condition Eq.~\eqref{thm1-1} holds, and Theorem~2
gives the asserted no-go result.

For comparison, Yao \textit{et al.}\cite{yao2025protocols} proved that for two depolarized copies
the unrestricted optimal universal purification protocol is
$\cE_{\mathrm{sym}}=\Tr_2\circ\mathcal M_2$, where
\begin{equation}
 \mathcal M_2(\rho)=P_{\mathrm{sym}}\rho P_{\mathrm{sym}},
 \qquad P_{\mathrm{sym}}=\frac{I+S}{2}, \qquad \text{$S$ is SWAP operator},
 \label{eq:n2c-sym-protocol}
\end{equation}
and their formulas give Eq.~\eqref{eq:n2c-unrestricted}
\cite{yao2025protocols}.  The operation $\mathcal M_2$ defining
this protocol is not energy-preserving for any
Hamiltonian in the family.  Indeed, an operation
belongs to $\EPO(H)$ if and only if it admits Kraus operators $\{M_k\}_k$
satisfying $[M_k,H]=0$ for every $k$.  The operation $\mathcal M_2$ has
Choi rank one, so every Kraus representation consists of scalar multiples
of $P_{\mathrm{sym}}$.  However, we have
\begin{equation}
 [P_{\mathrm{sym}},H]\ne0.
 \label{eq:n2c-not-epo}
\end{equation}
Hence $\mathcal M_2\notin\EPO(H)$ and the protocol is not energy-preserving.  Finally,
\begin{equation}
 F_{\max}^{\mathrm{all}}(\cN_\gamma,2)
 -F_{\text{raw}}(\cN_\gamma)
 =\frac{(1+\gamma)(3+\gamma)}{2(3+\gamma^2)}-\frac{(1+\gamma)}{2}=\frac{\gamma(1-\gamma^2)}{2(3+\gamma^2)}>0,
 \label{eq:n2c-gap}
\end{equation}
which shows that the protocol $\cE_{\mathrm{sym}}=\Tr_2\circ\mathcal M_2$ can indeed increase the average fidelity.
\end{proof}

This proposition supplies a concrete and physically meaningful family for
which the necessary and sufficient criterion of Theorem~2 is actually
met. Although two noisy copies contain enough information for unrestricted operations to increase the fidelity, the required operations inevitably disturb the energy status, implying the necessity of consuming external energy resources.
The family thus exhibits a strict separation between
energy-preserving and unrestricted purification, directly demonstrates
how an energetic restriction can completely eliminate a purification advantage that is
otherwise available. 

\section{Optimal universal energy-preserving purification}
Recently, Yao \textit{et al.}~\cite{yao2025protocols} revealed that the trade-off between average success probability and average fidelity exhibits a saturation behavior: sacrificing probability enhances fidelity only up to a critical threshold, beyond which the fidelity stabilizes at its global maximum. This saturation constitutes the \textit{golden point} of purification, representing the ideal balance where maximum fidelity is achieved with the highest possible yield. Capturing this insight, we define the \textit{optimal universal energy-preserving purification protocol} as the one that attains this golden point under zero energy cost constraints. This dictates a strict hierarchy of optimization: the protocol must yield the maximum average success probability conditioned on achieving the global maximum average fidelity allowed by energy-preserving operations. This definition ensures the protocol delivers the optimal purification performance while retaining the highest possible efficiency.
%In Yao's work~\cite{yao2025protocols}, they observed an interesting trade-off phenomenon that by tuning down the average success probability $\Bar{p}$, the corresponding average fidelity $\Bar{F}$ increases. But when $\Bar p$ down below the threshold $\Bar{p}^*$, $\Bar{F}$ remains and no longer increases. They called the turning point the \textit{gold point} of purification. Therefore, we define that for a given quantum noise $\cN$, Hamiltonian $H$, an $n\rightarrow 1$ universal energy-preserving purification protocol $\cE^*$ is optimal if its average fidelity $\Bar{F}(\cN,\cE^*,H,n)$ and the average success probability  $\Bar{p}(\cN,\cE^*,H,n)$ are at golden point.

%%%%%%%%%%%%%%%%%%%%%%%%%%%%%%%%%%%%%%%%%%%

%In the Stinespring representation of $\Lambda \in \text{EPO}(H)$, the conditions $[U, H \otimes I_{\mathcal{H}_{\text{env}}}] = [U, I_{\mathcal{H}} \otimes H_{\text{env}}] = 0$ and $[Q_{\text{env}}, H_{\text{env}}] = 0$ guarantee that the unitary $U$ and the projective measurement $\{Q_{\text{env}}, I_{\mathcal{H}_{\text{env}}}-Q_{\text{env}}\}$ are all energy-preserving for both the system and the environment.
%%%%%%%%%%%%%%%%%%%%%%%%%%%%%%%%%%%%%%%%%%%%%%%%%%%%%%%%%%%%%%%%%%%%%%%%%%%

Formally, we implement this optimization hierarchy via a two-stage maximization. First, we determine the maximum average fidelity across the entire space of universal energy-preserving purification protocols, $\text{UEPP}(\mathcal{N}, H, n)$:
\begin{equation}
    F_{\text{max}}(\mathcal{N}, H, n) \coloneqq \max_{\mathcal{E} \in \text{UEPP}(\mathcal{N}, H, n)} \bar{F}(\mathcal{N}, \mathcal{E}, n).
\end{equation}
The protocols achieving this maximum constitute the optimal fidelity set $\text{UEPP}^*(\mathcal{N}, H, n)\subseteq\text{UEPP}(\mathcal{N}, H, n)$:
\begin{equation}
\text{UEPP}^*(\mathcal{N}, H, n)\coloneqq\{\cE\in\text{UEPP}(\cN,H,n):\bar{F}(\mathcal{N}, \mathcal{E}, n)=F_{\text{max}}(\mathcal{N}, H, n)\}.
\end{equation}
Crucially, this set is not a singleton; it contains protocols that yield the identical global maximum fidelity but possess varying average purification success probabilities. Consequently, the second step identifies the optimal protocol $\mathcal{E}^*$ by maximizing the average success probability within this subset:
\begin{equation}
p_{\text{succ}}^{\text{max}}(\mathcal{N}, H, n) \coloneqq \max_{\mathcal{E} \in \text{UEPP}^*(\mathcal{N}, H, n)} \bar{p}(\mathcal{N}, \mathcal{E}, n).
\end{equation}
With the optimization framework established, the main goal is to determine the optimal protocol $\cE^*$ and access its performance, $F_{\text{max}}(\mathcal{N}, H, n)$ and $p_{\text{succ}}^{\text{max}}(\mathcal{N}, H, n)$. We begin by analytically deriving the maximum average fidelity.

\begin{theorem}
\label{thm2}
Let $\cN$ satisfy Eq.~\eqref{eq:noise-condition}, and let $H$ be a given Hamiltonian. Then the maximum average fidelity of $n\rightarrow 1$ universal energy-preserving purification protocols is given by
\begin{equation}
    F_{\text{max}}(\mathcal{N}, H, n) = \left\lVert C^{-\frac{1}{2}} A C^{-\frac{1}{2}} \right\rVert_\infty,
\end{equation}
where $A$ and $C$ are defined as in Theorem \ref{thm1} encapsulating the information of $(\mathcal{N}, H, n)$, and $\left\lVert \cdot \right\rVert_\infty$ is the spectral norm.
\end{theorem}

Theorem~\ref{thm2} identifies the fundamental limit of purification fidelity attainable under strict energy conservation, which is directly parametrized by the noise and energy constraint. Physically, this analytical maximum quantifies the ultimate reversibility of the noise process when no external work is permitted. Beyond the fidelity value, our analysis indeed enables a complete description of all energy-preserving purification protocols (Appendix~\ref{supp_sec:C}).

% In  we generalize Theorem~\ref{thm2}, which explicitly parametrizes the entire set $\text{UEPP}(\mathcal{N}, H, n)$ and specifically the optimal fidelity set $\text{UEPP}^*(\mathcal{N}, H, n)$, thereby identifying all protocols that achieve the optimal fidelity.

After identifying the optimal fidelity, the next step is to find the protocol within the optimal fidelity set that maximizes the purification yield. This search can be formulated as the following solvable semidefinite program (SDP).

\begin{proposition}
\label{pro3}
Let $\cN$ satisfy Eq.~\eqref{eq:noise-condition}, and let $H$ be a given Hamiltonian. Then the maximum average success probability of $n\rightarrow 1$ universal energy-preserving purification protocols can be computed via the following semidefinite program (SDP):
\begin{align}
p_{\text{succ}}^{\text{max}}(\mathcal{N}, H, n)=\textrm{maximize} \quad & \operatorname{Tr}(\Gamma_{\Lambda} C) \label{eq:pro 3-1}\\
\textrm{subject to} \quad & \Gamma_{\Lambda} \geq 0, \label{eq:pro 3-2}\\
& \operatorname{Tr}_{\text{out}}(\Gamma_{\Lambda}) \leq I_{\mathcal{H}^{\otimes n}}, \label{eq:pro 3-3}\\
 C^{-\frac{1}{2}} P_m C^{\frac{1}{2}} \cdot &\Gamma_{\Lambda} \cdot C^{\frac{1}{2}} P_m C^{-\frac{1}{2}} = \Gamma_{\Lambda}, \label{eq:pro 3-4}
\end{align}
where $C$ and $P_m$ are defined as in Theorem \ref{thm1}, $\operatorname{Tr}_{\text{out}}$ is the partial trace over the output system. Moreover, the solution $\Gamma_{\Lambda^*}$ of the SDP represents the Choi operator of $\Lambda^*$, yielding the optimal protocol $\cE^*=\operatorname{Tr}_{2\cdots n} \circ \Lambda^*$ that captures both $F_{\text{max}}(\mathcal{N}, H, n)$ and $p_{\text{succ}}^{\text{max}}(\mathcal{N}, H, n)$.
\end{proposition}

\begin{proof}
(i). First we show that for any $\cE\in\text{UEPP}^*(\cN, H, n)$, there exists $\Gamma_\Lambda$ with Eq.~\eqref{eq:pro 3-2} - Eq.~\eqref{eq:pro 3-4}, such that $\Bar{p}(\cN,\cE,n)=\operatorname{Tr}(\Gamma_\Lambda  C)$: Write $\cE = \operatorname{Tr}_{2\cdots n} \circ \Lambda$ for some $\Lambda\in\text{EPO}(H)$, denote the Choi operator of $\Lambda$ as $\Gamma_\Lambda$. Then it is clear that $\Gamma_\Lambda$ satisfies Eq.~\eqref{eq:pro 3-2} - Eq.~\eqref{eq:pro 3-3}. Note that by the proof of Theorem~\ref{thm1} and Theorem~\ref{thm6}, we can express $\Bar{p}(\cN,\cE,n)$ as
\begin{equation}
\Bar{p}(\cN,\cE,n)=\int \text{d}\psi \, \operatorname{Tr}[\mathcal{E}(\mathcal{N}(\psi)^{\otimes n})]=\operatorname{Tr}(\Gamma_\Lambda  C)> 0.
\end{equation}
Similarly, the average fidelity can be written as 
\begin{align}
\Bar{F}(\cN,\cE,n)&=\frac{\int \text{d}\psi \, F(\cE(\cN(\psi)^{\otimes n}), \psi)}{\Bar{p}(\cN,\cE,n)}\\
&=\frac{\operatorname{Tr}(\Gamma_\Lambda\cdot A)}{\operatorname{Tr}(\Gamma_\Lambda\cdot C)}\\
&=\operatorname{Tr}[\frac{C^{\frac{1}{2}}\cdot\Gamma_\Lambda\cdot C^{\frac{1}{2}}}{\operatorname{Tr}(\Gamma_\Lambda\cdot C)}\cdot C^{-\frac{1}{2}} A C^{-\frac{1}{2}}].
\end{align}
Then by the condition that $\cE\in\text{UEPP}^*(\cN, H, n)$ and $F_{\text{max}}(\mathcal{N}, H, n) = \left\lVert C^{-\frac{1}{2}} A C^{-\frac{1}{2}} \right\rVert_\infty$ (see Theorem~\ref{thm2}), the state $\frac{C^{\frac{1}{2}}\cdot\Gamma_\Lambda\cdot C^{\frac{1}{2}}}{\operatorname{Tr}(\Gamma_\Lambda\cdot C)}$ must be supported in the eigenspace of $C^{-\frac{1}{2}} A C^{-\frac{1}{2}}$ with maximum eigenvalue, that is,
\begin{equation}
P_m\frac{C^{\frac{1}{2}}\cdot\Gamma_{\Lambda}\cdot C^{\frac{1}{2}}}{\operatorname{Tr}(\Gamma_{\Lambda}\cdot C)}P_m=\frac{C^{\frac{1}{2}}\cdot\Gamma_{\Lambda}\cdot C^{\frac{1}{2}}}{\operatorname{Tr}(\Gamma_{\Lambda}\cdot C)},
\end{equation}
which is equivalent to Eq.~\eqref{eq:pro 3-4} for $C^{-\frac{1}{2}}\cdot C^{\frac{1}{2}}=C^{\frac{1}{2}}\cdot C^{-\frac{1}{2}}=\Pi$ and $\Pi\cdot\Gamma_\Lambda\cdot\Pi=\Gamma_\Lambda$, with $\Pi = \sum_E P_E^T \otimes P_E$ constructed from the spectral decomposition $H=\sum_E E P_E$. Thus we have $p_{\text{succ}}^{\text{max}}(\mathcal{N}, H, n)\leq\max\limits_{\Gamma_\Lambda}\operatorname{Tr}(\Gamma_\Lambda  C)$, where the maximization is among all $\Gamma_\Lambda$ with conditions Eq.~\eqref{eq:pro 3-2} - Eq.~\eqref{eq:pro 3-4};
\\
\\
(ii). Now we show that for any $\Gamma_\Lambda$ with Eq.~\eqref{eq:pro 3-2} - Eq.~\eqref{eq:pro 3-4}, we must have $\operatorname{Tr}(\Gamma_\Lambda  C)\leq p_{\text{succ}}^{\text{max}}(\mathcal{N}, H, n)$: Note that any $\Gamma_\Lambda$ with Eq.~\eqref{eq:pro 3-2} - Eq.~\eqref{eq:pro 3-4} must be a Choi operator corresponding to $\Lambda\in\text{EPO}(H)$, for 
\begin{equation}
\Pi\cdot\Gamma_\Lambda\cdot\Pi=\Pi\cdot C^{-\frac{1}{2}} P_m C^{\frac{1}{2}} \cdot \Gamma_{\Lambda} \cdot C^{\frac{1}{2}} P_m C^{-\frac{1}{2}}\cdot\Pi=C^{-\frac{1}{2}} P_m C^{\frac{1}{2}} \cdot \Gamma_{\Lambda} \cdot C^{\frac{1}{2}} P_m C^{-\frac{1}{2}}=\Gamma_\Lambda.
\end{equation}
Furthermore, we assume that $\operatorname{Tr}(\Gamma_\Lambda  C)>0$, since $\operatorname{Tr}(\Gamma_\Lambda  C)\leq p_{\text{succ}}^{\text{max}}(\mathcal{N}, H, n)$ if $\operatorname{Tr}(\Gamma_\Lambda  C)=0$. Consider $\cE \coloneqq \operatorname{Tr}_{2\cdots n} \circ \Lambda$, then clearly we have $\int \text{d}\psi \, \operatorname{Tr}[\mathcal{E}(\mathcal{N}(\psi)^{\otimes n})]= \operatorname{Tr}(\Gamma_\Lambda  C)>0$. Write
\begin{equation}
\frac{\int \text{d}\psi \, F(\cE(\cN(\psi)^{\otimes n}), \psi)}{\int \text{d}\psi \, \operatorname{Tr}[\mathcal{E}(\mathcal{N}(\psi)^{\otimes n})]}=\operatorname{Tr}[\frac{C^{\frac{1}{2}}\cdot\Gamma_\Lambda\cdot C^{\frac{1}{2}}}{\operatorname{Tr}(\Gamma_\Lambda\cdot C)}\cdot C^{-\frac{1}{2}} A C^{-\frac{1}{2}}].
\end{equation}
Note that $\Gamma_\Lambda$ satisfies the condition Eq.~\eqref{eq:pro 3-4}, i.e., $P_m\frac{C^{\frac{1}{2}}\cdot\Gamma_\Lambda\cdot C^{\frac{1}{2}}}{\operatorname{Tr}(\Gamma_\Lambda\cdot C)}P_m=\frac{C^{\frac{1}{2}}\cdot\Gamma_\Lambda\cdot C^{\frac{1}{2}}}{\operatorname{Tr}(\Gamma_\Lambda\cdot C)}$. This implies that 
\begin{equation}
\frac{\int \text{d}\psi \, F(\cE(\cN(\psi)^{\otimes n}), \psi)}{\int \text{d}\psi \, \operatorname{Tr}[\mathcal{E}(\mathcal{N}(\psi)^{\otimes n})]}=\operatorname{Tr}[\frac{C^{\frac{1}{2}}\cdot\Gamma_\Lambda\cdot C^{\frac{1}{2}}}{\operatorname{Tr}(\Gamma_\Lambda\cdot C)}\cdot C^{-\frac{1}{2}} A C^{-\frac{1}{2}}] = \left\lVert C^{-\frac{1}{2}} A C^{-\frac{1}{2}} \right\rVert_\infty=F_{\text{max}}(\mathcal{N}, H, n),
\end{equation}
hence $\cE\in\text{UEPP}^*(\cN, H, n)$ and we have $\operatorname{Tr}(\Gamma_\Lambda  C)=\Bar{p}(\cN,\cE,n)\leq p_{\text{succ}}^{\text{max}}(\mathcal{N}, H, n)$.
\\
\\
By the discussion before, we can compute $p_{\text{succ}}^{\text{max}}(\mathcal{N}, H, n)$ by the SDP Eq.~\eqref{eq:pro 3-1} - Eq.~\eqref{eq:pro 3-4}. Denote the solution of the SDP as $\Gamma_{\Lambda^*}$, then it represents a Choi operator of $\Lambda^*\in\text{EPO}(H)$, and $\cE^* = \operatorname{Tr}_{2\cdots n} \circ \Lambda^*$ is the optimal universal energy-preserving purification protocol for delivering both $F_{\text{max}}(\mathcal{N}, H, n)$ and $p_{\text{succ}}^{\text{max}}(\mathcal{N}, H, n)$.
\end{proof}

Proposition~\ref{pro3} recasts the second stage of the optimization as a finite-dimensional SDP, yielding both the maximum average success probability among fidelity-maximizing protocols and the Choi operator of an optimal protocol. Although constructive, this formulation is not generally efficient for large $n$, because the dimension grows exponentially with the number of copies. This motivates identifying cases in which the optimal purification can be obtained analytically.
% The proof of Proposition~\ref{pro3} is provided in Appendix~\ref{supp_sec:D}.

\begin{corollary}\label{cor:rank-one}
If $P_m$ has rank one for its setting $(\mathcal{N},H,n)$, then the optimal $n\rightarrow 1$ universal energy-preserving purification protocol $\cE^* = \operatorname{Tr}_{2\cdots n} \circ \Lambda^*$ is characterized by $\Lambda^*$ with Choi operator
\begin{equation}
\Gamma_{\Lambda^*}=\frac{C^{-1/2}P_mC^{-1/2}}{ \|\operatorname{Tr}_{\rm out}C^{-1/2}P_mC^{-1/2}\|_\infty}.
\end{equation}
And the corresponding maximum average success probability is
\begin{equation}
p_{\text{succ}}^{\text{max}}(\mathcal{N}, H, n)=\|\operatorname{Tr}_{\rm out}C^{-1/2}P_mC^{-1/2}\|_\infty^{-1}.
\end{equation}
\end{corollary}

\begin{proof}
For every $\Gamma_\Lambda$ that delivers the maximum average fidelity, Eq.~\eqref{eq:pro 3-4} leads to
\begin{equation}
P_m C^{1/2}\Gamma_\Lambda C^{1/2}P_m=C^{1/2}\Gamma_\Lambda C^{1/2}.
\end{equation}
Since $P_m$ has rank one, we must have $C^{1/2}\Gamma_\Lambda C^{1/2}=qP_m$ for some scalar $q\geq0$. Using $C^{-1/2}C^{1/2}=\Pi$ and $\Pi\Gamma_\Lambda\Pi=\Gamma_\Lambda$ gives
\begin{equation}
\Gamma_\Lambda=qC^{-1/2}P_mC^{-1/2}.
\end{equation}
Then the trace non-increasing condition is equivalent to
\begin{equation}
q\,\operatorname{Tr}_{\rm out}(C^{-1/2}P_mC^{-1/2})\leq I,
\end{equation}
which is precisely $q\leq\|\operatorname{Tr}_{\rm out}(C^{-1/2}P_mC^{-1/2})\|_\infty^{-1}$. Moreover, the average success probability of $\Lambda$ reads $\operatorname{Tr}(\Gamma_\Lambda C)=q\operatorname{Tr}(P_m)=q$. Maximizing the average success probability $q$ therefore yields the two displayed formulas.
\end{proof}

Corollary~\ref{cor:rank-one} provides a simplification of Proposition~\ref{pro3} in a specific case. When $P_m$ has rank one, the fidelity-maximizing Choi operator is fixed up to an overall scalar, whose largest admissible value is determined by the trace non-increasing condition. The optimal protocol and its maximum average success probability therefore follow in closed form, without solving the second-stage SDP.

Remarkably, when the Hamiltonian $H$ is fully degenerate, $\text{EPO}(H)$ recovers all CPTN. In this case, the energy-preserving restriction vanishes, and the energy-preserving purification scenario reduces to the ordinary one without resource constraints, which is studied in Ref.~\cite{fiuravsek2004optimal}. All results presented here therefore apply directly to that ordinary setting by specifying $H$ to be fully degenerate. For example, Theorem \ref{thm2} reduces to the result of Ref.~\cite{fiuravsek2004optimal} when $H=I_{\mathcal{H}^{\otimes n}}$. 

\section{Implementation of energy-preserving purification}
While Proposition~\ref{pro3} yields the Choi characterization of the optimal purification operation $\Lambda^*$, its physical realization demands a concrete energy-free unitary transformation and measurement. Given that $\Lambda^*$ represents only the successful branch of the probabilistic scheme, a zero energy cost implementation, however, requires a complete energy-preserving instrument. We therefore construct a complementary failure operation $\Lambda^\#$ such that $\Lambda^*+\Lambda^\#$ is an energy-preserving channel.

%While Proposition~\ref{pro3} yields the Choi characterization of the optimal successful operation $\Lambda^*$, a physical realization requires a complete energy-preserving instrument. We therefore construct a complementary failure operation $\Lambda^\#$ such that $\Lambda^*+\Lambda^\#$ is an energy-preserving channel.

\begin{lemma}
\label{lem2}
Let $H$ be a Hamiltonian and let $\Lambda^*\in\mathrm{EPO}(H)$ have Kraus operators $\{M_k\}_k$ satisfying $[M_k,H]=0$. Define
\begin{equation}
R:=I-\sum_k M_k^\dagger M_k,
\qquad
\Lambda^\#(\rho):=R^{1/2}\rho R^{1/2}.
\end{equation}
Then $\Lambda^\#\in\mathrm{EPO}(H)$ and $\Lambda^*+\Lambda^\#$ is an energy-preserving channel for $H$.
\end{lemma}

\begin{proof}
Since $\Lambda^*$ is trace non-increasing, we have $R\geq0$. Furthermore, each $M_k$ commutes with $H$, so $[R,H]=[R^{1/2},H]=0$. The Kraus criterion in Lemma~\ref{lem1} therefore gives $\Lambda^\#\in\mathrm{EPO}(H)$. Finally,
\begin{equation}
\sum_kM_k^\dagger M_k+R=I,
\end{equation}
so $\Lambda^*+\Lambda^\#$ is trace preserving; all of its Kraus operators commute with $H$, and it is thus an energy-preserving channel.
\end{proof}

% The detailed construction of an energy-preserving instrument is presented in Appendix~\ref{supp_sec:E} for completeness.

% The following lemma explicitly constructs this complement to complete the energy-preserving instrument.

% \begin{lemma}
% \label{lem2}
% Let $H$ be a Hamiltonian, and let $\Lambda^* \in \text{EPO}(H)$. Fix an orthonormal eigenbasis $\{\ket{i}\}_i$ of $H$. Then the operator
% \[
% \Gamma_{\Lambda^\#} := \sum_i \ket{i}\bra{i} \otimes \rho_i,
% \]
% where
% \[
% \rho_i = 
% \begin{cases}
% \dfrac{\Lambda^*(\ket{i}\bra{i})}{\operatorname{Tr}[\Lambda^*(\ket{i}\bra{i})]} - \Lambda^*(\ket{i}\bra{i}), & \operatorname{Tr}[\Lambda^*(\ket{i}\bra{i})] \neq 0, \\ \\
% \ket{i}\bra{i}, & \text{otherwise},
% \end{cases}
% \]
% is the Choi operator of an operation $\Lambda^\# \in \text{EPO}(H)$ such that $\Lambda^* + \Lambda^\#$ is an energy-preserving channel for $H$.
% \end{lemma}

Once constructed energy-preserving instrument $\{\Lambda^*,\ \Lambda^\# \}$ from $\Gamma_{\Lambda^*}$, one can follow the routine in Ref.~\cite{chiribella2017optimal} to obtain the physical implementation, including the energy-free unitary transformation, and the energy-free projective measurements. The procedure is presented here for completeness.

First, without loss of generality let  $M_1,\cdots,M_r,M_{r+1},\cdots,M_m$ be Kraus operators of $\Lambda^*+\Lambda^\#$, where $M_1,\cdots,M_r$ and $M_{r+1},\cdots,M_m$ are Kraus operators of $\Lambda^*$ and $\Lambda^\#$, respectively. We have $[M_k,H]=0$ for all $k$.

Second, construct the environment's space $\mathcal{H}_{\text{env}}$ as $m$-dimensional Hilbert space and specify its orthonormal basis $\{\ket{\phi_1},\cdots,\ket{\phi_m}\}$. Set the environment's Hamiltonian $H_{\text{env}}$ as $I_{\mathcal{H}_{\text{env}}}$. Then $\ket{\phi_1}$ is naturally the ground state of $H_{\text{env}}$. Moreover, $[U,I\otimes H_{\text{env}}]=0$ for any unitary $U$ and $[Q_{\text{env}},H_{\text{env}}]=0$ for any projector $Q_{\text{env}}$. Finally, set the unitary and projective measurement. Let

\begin{equation}
H=\sum\limits_EE P_E=\sum\limits_{E}E\sum\limits_{\substack{i\\ P_E\ket{i}=\ket{i}}}\ket{i}\bra{i}
\end{equation} 
be the spectral decomposition of system's Hamiltonian $H$, where $E$ denotes eigenvalue of $H$, $P_E$ denotes the projector on the eigensubspace $S_E$ of $H$ with respect to $E$, $\{\ket{i}\}_i$ forms an orthonormal eigenbasis of $H$ and $\{\ket{i}:P_E\ket{i}=\ket{i}\}$ constitutes an orthonormal basis of $S_E$. Then

\begin{equation}
H\otimes I_{\mathcal{H}_{\text{env}}}=\sum\limits_{E}E\sum\limits_{\substack{i,k\\ P_E\ket{i}=\ket{i}}}\ket{i}\otimes\ket{\phi_k}\cdot\bra{i}\otimes\bra{\phi_k}
\end{equation}
is the spectral decomposition of $H\otimes I_{\mathcal{H}_{\text{env}}}$. Now $\{\ket{i}\otimes\ket{\phi_k}\}_{i,k}$ is orthonormal eigenbasis of $H\otimes I_{\mathcal{H}_{\text{env}}}$ and $B_E:=\{\ket{i}\otimes\ket{\phi_k}:P_E\ket{i}=\ket{i},k\}$ constitutes an orthonormal basis of the eigensubspace $S_E\otimes\mathcal{H}_{\text{env}}$ of $H\otimes I_{\mathcal{H}_{\text{env}}}$ with respect to $E$. Note that $\{\sum\limits_kM_k\ket{i}\otimes\ket{\phi_k}\}_i$ is an orthonormal set of eigenvectors of $H\otimes I_{\mathcal{H}_{\text{env}}}$ and $\{\sum\limits_kM_k\ket{i}\otimes\ket{\phi_k}:P_E\ket{i}=\ket{i}\}\subseteq S_E\otimes\mathcal{H}_{\text{env}}$ for each $E$. Then we extend $\{\sum\limits_kM_k\ket{i}\otimes\ket{\phi_k}:P_E\ket{i}=\ket{i}\}$ to an orthonormal basis $B_E'$ of $S_E\otimes\mathcal{H}_{\text{env}}$ for each $E$. Now we characterize the unitary $U$ by mapping $B_E$ to $B_E'$ bijectively for each $E$, while satisfying $U\ket{i}\otimes\ket{\phi_1}=\sum\limits_kM_k\ket{i}\otimes\ket{\phi_k}$ for any $i$. Then it is clear that $[U,H\otimes I_{\mathcal{H}_{\text{env}}}]=0$. Recall that $M_1,\cdots,M_r$ are Kraus operators of $\Lambda^*$, consider the projector $Q_{\text{env}}:=\sum\limits_{k=1}^{r}\ket{\phi_k}\bra{\phi_k}$ and the projective measurement $\{Q_{\text{env}},I-Q_{\text{env}}\}$, then we obtain
\begin{equation}
\Lambda^*(\rho)=\operatorname{Tr}_{\mathcal{H}_\text{env}}[I\otimes Q_{\text{env}}\cdot U(\rho\otimes\ket{\phi_1}\bra{\phi_1})U^\dagger]
\end{equation}
for any state $\rho$, which characterizes the physical implementation of $\Lambda^*$.

% The detailed procedures are presented in Appendix~\ref{supp_sec:F} for completeness. 

% With the Choi operator $\Gamma_{\Lambda^*}$ of the optimal energy-preserving purification operation solved by the SDP, our next goal is to specify the physical implementation of the  purification protocol $\Gamma_{\Lambda^*}$, which is operated without energy cost at each step and consists of the following three components: 
% \begin{itemize}
%     \item an environment's system $\mathcal{H}_{\text{env}}$ prepared in the ground state $\ket{\phi_1}$ of an environment's Hamiltonian $H_{\text{env}}$;
%     \item a system--environment unitary $U$ satisfying
%     $[U, H \otimes I_{\mathcal{H}_{\text{env}}}] = [U, I \otimes H_{\text{env}}] = 0$;
%     \item a projective measurement $\{Q_{\text{env}}, I - Q_{\text{env}}\}$ on the environment such that $[Q_{\text{env}}, H_{\text{env}}] = 0$.
% \end{itemize}
% These elements together realize the operation as
% \[
% \Lambda^*(\rho) = \operatorname{Tr}_{\mathcal{H}_\text{env}}\left[ (I \otimes Q_{\text{env}})\cdot \, U (\rho \otimes \ket{\phi_1}\bra{\phi_1}) U^\dagger \right], \quad \forall \rho.
% \]

\section{Numerical experiments}
Here, we conduct numerical experiments to demonstrate the effectiveness of the optimal universal energy-preserving purification protocol. Specifically, we consider the depolarizing noise $\cN$ with noise coefficients ranging from $\gamma\in[0.05,1]$, and the transverse-field
Ising Hamiltonian with open boundary conditions:
\begin{equation}
    H
    =J\sum_{i=1}^{N-1} Z_i Z_{i+1}
    +h\sum_{i=1}^{N} X_i,
\end{equation}
where $Z$ and $X$ refer to the Pauli-$Z$ and Pauli-$X$ matrices, respectively. $J$ and $h$ are the coupling
constants, which are set as $J=-1$ and $h=-1.5$. The optimal fidelity achieved by the protocol with respect to such a Hamiltonian $H$, is shown in Fig.~\ref{fig:results} (a), and the corresponding optimal success probability is illustrated in Fig.~\ref{fig:results} (b). In the numerical experiments, we consider the $2\rightarrow1$ and $3\rightarrow 1$ purification scenarios, and we found that the $3\rightarrow 1$ protocol achieves higher fidelity than the $2\rightarrow 1$ protocol. Such a result is intuitive as the $3\rightarrow 1$ protocol consumes more resources. Note that the success probability of $3\rightarrow 1$ protocol is lower than that of the $2\rightarrow 1$ protocol, which implies there is a trade-off between achievable fidelity and the success probability. In reality, we should choose the purification protocol carefully by considering both the fidelity and the success probability.

\begin{figure}[h]
    \centering 
    \includegraphics[width=\textwidth]{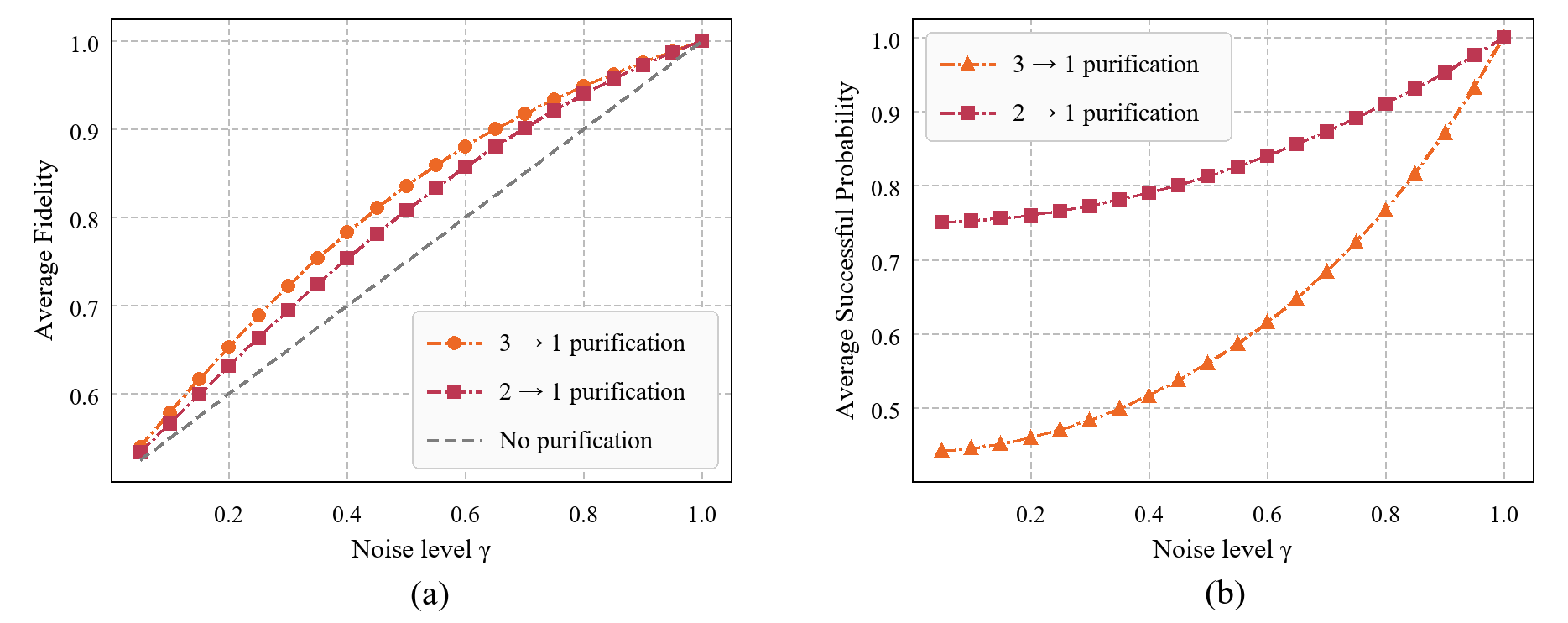}
    \caption{(a) Optimal average fidelity by energy-preserving purification. (b) Optimal average success probability by energy-preserving purification.}
    \label{fig:results} 
\end{figure}

\section{Conclusion}

In this work, we established a framework for universal quantum state purification under exact energy-preserving operations. For noise channels satisfying the full-support condition in Eq.~\eqref{eq:noise-condition}, including nontrivial depolarizing noises, we derived a necessary and sufficient criterion for when no universal energy-preserving purifier can outperform the unpurified input. A concrete family of nondegenerate two-qubit Hamiltonians shows that this restriction can completely remove a fidelity gain attainable by unrestricted operations. When purification remains possible, we obtained the maximum average fidelity, characterized the fidelity-maximizing protocols, and formulated a finite-dimensional SDP for the maximum average success probability at that fidelity, with a closed-form solution in the rank-one case. We also constructed an energy-preserving instrument and dilation implementing the optimal successful branch, while numerical examples illustrate the fidelity--success trade-off for two and three input copies. For a fully degenerate Hamiltonian, the framework reduces to unconstrained probabilistic purification.

Together, these results show that universal purifiability is determined jointly by the noise, the number of copies, and the allowed operations: energy conservation is not merely an implementation detail, but can change whether purification is possible at all. The framework also admits a battery-assisted extension. For a specified external battery initialized in a full-rank state, requiring the joint system--battery operation to preserve total energy yields analogous no-go and optimality characterizations, including the achievable-fidelity range and the success-probability optimization. This extension, detailed in Appendix~\ref{supp_sec:G}, makes the energetic resource explicit and places zero-energy and energy-assisted purification within a unified framework.

\textit{Data availability}---The data that support the findings of this work are openly available~\cite{codes}.

\textbf{Acknowledgments}---
X.W. acknowledges the support from the National Key R\&D Program of China (Grant No. 2024YFB4504004), the National Natural Science Foundation of China (Grant. No. 12447107), the Guangdong Provincial Quantum Science Strategic Initiative (Grant No. GDZX2403008, GDZX2503001).

% , the Guangdong Provincial Key Lab of Integrated Communication, Sensing and Computation for Ubiquitous Internet of Things (Grant No. 2023B1212010007).

% , the Quantum Science Center of Guangdong-Hong Kong-Macao Greater Bay Area.

% Finally, we add two remarks on the energy‑preserving purification framework, regarding its special case and the extension to the scenarios with a given engineered energy resource. 

% This connects the theory of quantum state purification with the practical imperative of energy efficiency in quantum technologies.

% The principal contributions are as follows. First, for a broad class of noise processes including non-trivial depolarizing noises, we have analytically derived the maximum achievable average fidelity under the energy-preserving constraint and identified the corresponding operations. Second, we have formulated the search for the optimal energy-preserving purification operation and its average success probability as a computationally tractable semidefinite program. The solution to this SDP routinely specifies a physical implementation composed entirely of steps that incur no energy cost. Third, we have provided a necessary and sufficient condition that dictates when any non-trivial purification is impossible without an energy cost, establishing a fundamental no-go theorem for energy-preserving purification.

%%%%%%%%%%%%%%%%%%%%%%%%%%%%%%%%%%%%%%%%%%%%%%%%%%%%%%%%%%%%
% Bibliography
%%%%%%%%%%%%%%%%%%%%%%%%%%%%%%%%%%%%%%%%%%%%%%%%%%%%%%%%%%%%
\bibliographystyle{apsrev4-2}
\bibliography{ref}

%%%%%%%%%%%%%%%%%%%%%%%%%%%%%%%%%%%%%%%%%%%%%%%%%%%%%%%%%%%%

\clearpage
\vspace{2cm}
\onecolumngrid
\vspace{2cm}
\begin{center}
{\textbf{\large Appendix for Universal quantum state purification under energy-preserving constraints }}
\end{center}
\appendix

%%%%%%%%%%%%%%%%%%%%%%%%%%%%%%%%%%%%%%%%%%%%%%%%%%%%%%%%%%%%
\renewcommand{\appendixname}{Appendix}

%%%%%%%%%%%%%%%%%%%%%%%%%%%%%%%%%%%%%%%%%%%%%%%%%%%%%%%%%%%%%%%%%%%%%%%%%%%
% \section{Energy-preserving operations}\label{supp_sec:A}

%%%%%%%%%%%%%%%%%%%%%%%%%%%%%%%%%%%%%%%%%%%%%%%%%%%%%%%%%%%%%%%%%%%%%%%%%%%

\section{Proof of Theorem~\ref{thm1}}\label{supp_sec:B}
\begin{proof}
\\
(i). To justify the direction $\Leftarrow$, it suffices to prove that
\begin{equation}
\frac{\int \text{d}\psi \,\operatorname{Tr}[\psi\cdot\cE(\mathcal{N}(\psi)^{\otimes n})]}{\int \text{d}\psi \,\operatorname{Tr}[\cE(\mathcal{N}(\psi)^{\otimes n})]}\leq\int \text{d}\psi \, \langle \psi | \mathcal{N}(\psi) | \psi \rangle
\end{equation}
holds for any energy-preserving $\cE$ with $\int \text{d}\psi \,\operatorname{Tr}[\cE(\mathcal{N}(\psi)^{\otimes n})]>0$. Or equivalently, 
\begin{equation}
\frac{\int \text{d}\psi \,\operatorname{Tr}[\psi\cdot\operatorname{Tr}_{2\cdots n}\circ\Lambda(\mathcal{N}(\psi)^{\otimes n})]}{\int \text{d}\psi \,\operatorname{Tr}[\Lambda(\mathcal{N}(\psi)^{\otimes n})]}\leq\int \text{d}\psi \, \langle \psi | \mathcal{N}(\psi) | \psi \rangle=\frac{\int \text{d}\psi \,\operatorname{Tr}[\psi\cdot\operatorname{Tr}_{2\cdots n}\circ id(\mathcal{N}(\psi)^{\otimes n})]}{\int \text{d}\psi \,\operatorname{Tr}[id(\mathcal{N}(\psi)^{\otimes n})]}
\end{equation}
holds for any $\Lambda\in\text{EPO}(H)$ with $\int \text{d}\psi \,\operatorname{Tr}[\Lambda(\mathcal{N}(\psi)^{\otimes n})]>0$.
\\
\\
\noindent Now for any such $\Lambda\in\text{EPO}(H)$, we have
\begin{align}
&\frac{\int \text{d}\psi \,\operatorname{Tr}[\psi\cdot\operatorname{Tr}_{2\cdots n}\circ\Lambda(\mathcal{N}(\psi)^{\otimes n})]}{\int \text{d}\psi \,\operatorname{Tr}[\Lambda(\mathcal{N}(\psi)^{\otimes n})]} \\
&=\frac{\int \text{d}\psi \,\operatorname{Tr}[\psi\otimes I_{\mathcal{H}^{\otimes (n-1)}}\cdot\Lambda(\mathcal{N}(\psi)^{\otimes n})]}{\int \text{d}\psi \,\operatorname{Tr}[\Lambda(\mathcal{N}(\psi)^{\otimes n})]} \\
&=\frac{\operatorname{Tr}[\Gamma_\Lambda\cdot\int \text{d}\psi \,(\mathcal{N}(\psi)^{\otimes n})^T\otimes\psi\otimes I_{\mathcal{H}^{\otimes (n-1)}}]}{\operatorname{Tr}[\Gamma_\Lambda\cdot (\int \text{d}\psi \,\mathcal{N}(\psi)^{\otimes n})^T\otimes I_{\mathcal{H}^{\otimes n}}]} \label{thm 1-1} \\
&=\frac{\operatorname{Tr}(\Gamma_\Lambda\cdot A)}{\operatorname{Tr}[(\Pi\cdot\Gamma_\Lambda\cdot \Pi)\cdot (\int \text{d}\psi \,\mathcal{N}(\psi)^{\otimes n})^T\otimes I_{\mathcal{H}^{\otimes n}}]} \label{thm 1-2}\\
&=\frac{\operatorname{Tr}(\Gamma_\Lambda\cdot A)}{\operatorname{Tr}(\Gamma_\Lambda\cdot C)}\label{thm 1-3}, 
\end{align}
where Eq.~\eqref{thm 1-1} - Eq.~\eqref{thm 1-2} utilizes the condition that $\Lambda\in\text{EPO}(H)$, that is, $\Gamma_\Lambda=\Pi\cdot\Gamma_\Lambda\cdot\Pi$ (see Lemma~\ref{lem1}).
\\
\\
%Since $\mathcal{N}$ is a non-trivial depolarizing noise, we can write $\cN(\psi) = (1-\gamma)\pi+ \gamma\psi$ for any pure state $\psi$, where $\gamma\in(0,1)$ and $\pi$ denotes the maximally mixed state. Then we have
%\begin{equation}\label{thm 1-4}
%\mathcal{N}(\psi)^{\otimes n}=(1-\gamma)^n\pi^{\otimes n}+[1-(1-\gamma)^n]\rho_\psi
%\end{equation}
%holds for any pure state $\psi$, where $\rho_\psi$ is a state depending on $\psi$. This leads to
%\begin{equation}\label{thm 1-5}
%\int \text{d}\psi \, \mathcal{N}(\psi)^{\otimes n}=(1-\gamma)^n\pi^{\otimes n}+[1-(1-\gamma)^n]\int \text{d}\psi \, \rho_\psi>0
%\end{equation}
%for $(1-\gamma)^n\pi^{\otimes n}>0$ and $[1-(1-\gamma)^n]\int \text{d}\psi \, \rho_\psi \geq0$.

By Eq.~\eqref{eq:noise-condition}, $\int d\psi\,\mathcal{N}(\psi)^{\otimes n}>0$. Since $(\int \text{d}\psi \,\mathcal{N}(\psi)^{\otimes n})^T\otimes I_{\mathcal{H}^{\otimes n}} >0$ and $\Pi$ is a projector, we have 
\begin{equation}\label{thm 1-7}
\text{Supp}(C)=\text{Supp}(\Pi\cdot[(\int \text{d}\psi \,\mathcal{N}(\psi)^{\otimes n})^T\otimes I_{\mathcal{H}^{\otimes n}}]\cdot\Pi)=\text{Supp}(\Pi).
\end{equation}
Subsequently,
\begin{equation}\label{thm 1-8}
C^{\frac{1}{2}}\cdot C^{-\frac{1}{2}}=C^{-\frac{1}{2}}\cdot C^{\frac{1}{2}}=\Pi.
\end{equation}
\\
Return to Eq.~\eqref{thm 1-3}, we have:
\begin{align}
\frac{\operatorname{Tr}(\Gamma_\Lambda\cdot A)}{\operatorname{Tr}(\Gamma_\Lambda\cdot C)}&=\frac{\operatorname{Tr}(\Pi\cdot\Gamma_\Lambda\cdot\Pi\cdot A)}{\operatorname{Tr}(\Gamma_\Lambda\cdot C)} \\
&=\frac{\operatorname{Tr}(C^{-\frac{1}{2}}\cdot C^{\frac{1}{2}}\Gamma_\Lambda C^{\frac{1}{2}}\cdot C^{-\frac{1}{2}}\cdot A)}{\operatorname{Tr}(\Gamma_\Lambda\cdot C)} \\
&=\operatorname{Tr}[\frac{C^{\frac{1}{2}}\cdot\Gamma_\Lambda\cdot C^{\frac{1}{2}}}{\operatorname{Tr}(\Gamma_\Lambda\cdot C)}\cdot C^{-\frac{1}{2}} A C^{-\frac{1}{2}}] \label{thm 4} \\
&\leq\left\lVert C^{-\frac{1}{2}} A C^{-\frac{1}{2}} \right\rVert_\infty. \label{thm 5} 
\end{align}
\\
Now from $C^{\frac{1}{2}}\cdot C^{-\frac{1}{2}}=C^{-\frac{1}{2}}\cdot C^{\frac{1}{2}}=\Pi$, $\text{Supp}(P_m)\subseteq\text{Supp}(C)=\text{Supp}(\Pi)$, and the condition Eq.~\eqref{thm1-1}, we have
\begin{equation}
P_m\frac{C^{\frac{1}{2}}\cdot\Gamma_{id}\cdot C^{\frac{1}{2}}}{\operatorname{Tr}(\Gamma_{id}\cdot C)}P_m=\frac{C^{\frac{1}{2}}\cdot\Gamma_{id}\cdot C^{\frac{1}{2}}}{\operatorname{Tr}(\Gamma_{id}\cdot C)},
\end{equation}
which means state $\frac{C^{\frac{1}{2}}\cdot\Gamma_{id}\cdot C^{\frac{1}{2}}}{\operatorname{Tr}(\Gamma_{id}\cdot C)}$ is supported on the eigenspace of $C^{-\frac{1}{2}} A C^{-\frac{1}{2}}$ with maximum eigenvalue $\left\lVert C^{-\frac{1}{2}} A C^{-\frac{1}{2}} \right\rVert_\infty$. So
\begin{equation}\label{thm 6}
\frac{\int \text{d}\psi \,\operatorname{Tr}[\psi\cdot\operatorname{Tr}_{2\cdots n}\circ id(\mathcal{N}(\psi)^{\otimes n})]}{\int \text{d}\psi \,\operatorname{Tr}[id(\mathcal{N}(\psi)^{\otimes n})]}=\frac{\operatorname{Tr}(\Gamma_{id}\cdot A)}{\operatorname{Tr}(\Gamma_{id}\cdot C)}=\operatorname{Tr}[\frac{C^{\frac{1}{2}}\cdot\Gamma_{id}\cdot C^{\frac{1}{2}}}{\operatorname{Tr}(\Gamma_{id}\cdot C)}\cdot C^{-\frac{1}{2}} A C^{-\frac{1}{2}}]=\left\lVert C^{-\frac{1}{2}} A C^{-\frac{1}{2}} \right\rVert_\infty.
\end{equation}
Combining Eq.~\eqref{thm 6} with Eq.~\eqref{thm 5}, the direction $\Leftarrow$ is then justified.
\\
\\
(ii). For the direction $\Rightarrow$, by the discussion of (i) we know that 
\begin{equation}\label{thm 7}
\frac{\int \text{d}\psi \,\operatorname{Tr}[\psi\cdot\operatorname{Tr}_{2\cdots n}\circ\Lambda(\mathcal{N}(\psi)^{\otimes n})]}{\int \text{d}\psi \,\operatorname{Tr}[\Lambda(\mathcal{N}(\psi)^{\otimes n})]}=\operatorname{Tr}[\frac{C^{\frac{1}{2}}\cdot\Gamma_\Lambda\cdot C^{\frac{1}{2}}}{\operatorname{Tr}(\Gamma_\Lambda\cdot C)}\cdot C^{-\frac{1}{2}} A C^{-\frac{1}{2}}]\leq\left\lVert C^{-\frac{1}{2}} A C^{-\frac{1}{2}} \right\rVert_\infty
\end{equation}
holds for any $\Lambda\in\text{EPO}(H)$ with $\int \text{d}\psi \,\operatorname{Tr}[\Lambda(\mathcal{N}(\psi)^{\otimes n})]>0$.\\ 
\\
Now for any state $\sigma$ whose support contained in the eigenspace of $C^{-\frac{1}{2}}AC^{-\frac{1}{2}}$ with maximum eigenvalue $\left\lVert C^{-\frac{1}{2}} A C^{-\frac{1}{2}} \right\rVert_\infty$, that is, any state $\sigma$ that satisfies $\sigma=P_m\sigma P_m$, we construct
\begin{equation}\label{pf-thm1-1}
\Gamma_{\Lambda_q}\coloneqq q C^{-\frac{1}{2}}\sigma C^{-\frac{1}{2}}
\end{equation}
for any $0 < q \leq \left\lVert \operatorname{Tr}_{\text{out}} \left( C^{-\frac{1}{2}} \sigma C^{-\frac{1}{2}} \right) \right\rVert_\infty^{-1}$ (note that $\left\lVert \operatorname{Tr}_{\text{out}} \left( C^{-\frac{1}{2}} \sigma C^{-\frac{1}{2}} \right) \right\rVert_\infty\neq0$ for any such $\sigma$ since $\operatorname{Tr}(C^{-\frac{1}{2}}\sigma C^{-\frac{1}{2}})=\operatorname{Tr}(C^{-1}\sigma)>0$). It is clear that $\Gamma_{\Lambda_q}\geq0$ and $\operatorname{Tr}_{\text{out}}(\Gamma_{\Lambda_q})\leq I$, which ensures $\Gamma_{\Lambda_q}$ are valid Choi operators for some CPTN $\Lambda_q$. Furthermore, note that $\Pi\cdot C^{-\frac{1}{2}}= C^{-\frac{1}{2}}\cdot\Pi= C^{-\frac{1}{2}}$, we have
\begin{equation}
\Pi\cdot\Gamma_{\Lambda_q}\cdot\Pi=\Gamma_{\Lambda_q},
\end{equation}
which means $\Gamma_{\Lambda_q}$ are valid Choi operators for some $\Lambda_q\in\text{EPO}(H)$ with $\int \text{d}\psi \,\operatorname{Tr}[\Lambda_q(\mathcal{N}(\psi)^{\otimes n})]=\operatorname{Tr}(\Gamma_{\Lambda_q}\cdot C)>0$. Then it is direct to verify that
\begin{equation}
P_m\frac{C^{\frac{1}{2}}\cdot\Gamma_{\Lambda_q}\cdot C^{\frac{1}{2}}}{\operatorname{Tr}(\Gamma_{\Lambda_q}\cdot C)}P_m=\frac{C^{\frac{1}{2}}\cdot\Gamma_{\Lambda_q}\cdot C^{\frac{1}{2}}}{\operatorname{Tr}(\Gamma_{\Lambda_q}\cdot C)},
\end{equation}
that is, $\frac{C^{\frac{1}{2}}\cdot\Gamma_{\Lambda_q}\cdot C^{\frac{1}{2}}}{\operatorname{Tr}(\Gamma_{\Lambda_q}\cdot C)}$ are states that are supported on the eigenspace of $C^{-\frac{1}{2}}AC^{-\frac{1}{2}}$ with maximum eigenvalue. So
\begin{equation}
\frac{\int \text{d}\psi \,\operatorname{Tr}[\psi\cdot\operatorname{Tr}_{2\cdots n}\circ\Lambda_q(\mathcal{N}(\psi)^{\otimes n})]}{\int \text{d}\psi \,\operatorname{Tr}[\Lambda_q(\mathcal{N}(\psi)^{\otimes n})]}=\operatorname{Tr}[\frac{C^{\frac{1}{2}}\cdot\Gamma_{\Lambda_q}\cdot C^{\frac{1}{2}}}{\operatorname{Tr}(\Gamma_{\Lambda_q}\cdot C)}\cdot C^{-\frac{1}{2}} A C^{-\frac{1}{2}}]=\left\lVert C^{-\frac{1}{2}} A C^{-\frac{1}{2}} \right\rVert_\infty.
\end{equation}
However, under the condition that there is no universal energy-preserving purification protocol for $(\mathcal{N},H,n)$, we must have
\begin{align}
\left\lVert C^{-\frac{1}{2}} A C^{-\frac{1}{2}} \right\rVert_\infty&=\frac{\int \text{d}\psi \,\operatorname{Tr}[\psi\cdot\operatorname{Tr}_{2\cdots n}\circ\Lambda_q(\mathcal{N}(\psi)^{\otimes n})]}{\int \text{d}\psi \,\operatorname{Tr}[\Lambda_q(\mathcal{N}(\psi)^{\otimes n})]}\\
&\leq\frac{\int \text{d}\psi \,\operatorname{Tr}[\psi\cdot\operatorname{Tr}_{2\cdots n}\circ id(\mathcal{N}(\psi)^{\otimes n})]}{\int \text{d}\psi \,\operatorname{Tr}[id(\mathcal{N}(\psi)^{\otimes n})]}\\
&=\operatorname{Tr}[\frac{C^{\frac{1}{2}}\cdot\Gamma_{id}\cdot C^{\frac{1}{2}}}{\operatorname{Tr}(\Gamma_{id}\cdot C)}\cdot C^{-\frac{1}{2}} A C^{-\frac{1}{2}}] \label{thm 8}\\
&\leq\left\lVert C^{-\frac{1}{2}} A C^{-\frac{1}{2}} \right\rVert_\infty \label{thm 9},
\end{align}
where Eq.~\eqref{thm 8} - Eq.~\eqref{thm 9} is from Eq.~\eqref{thm 7}. Then 
\begin{equation}
\operatorname{Tr}[\frac{C^{\frac{1}{2}}\cdot\Gamma_{id}\cdot C^{\frac{1}{2}}}{\operatorname{Tr}(\Gamma_{id}\cdot C)}\cdot C^{-\frac{1}{2}} A C^{-\frac{1}{2}}]=\left\lVert C^{-\frac{1}{2}} A C^{-\frac{1}{2}} \right\rVert_\infty.
\end{equation}
So $\frac{C^{\frac{1}{2}}\cdot\Gamma_{id}\cdot C^{\frac{1}{2}}}{\operatorname{Tr}(\Gamma_{id}\cdot C)}$ must be a state whose support is contained in the eigenspace of $ C^{-\frac{1}{2}} A C^{-\frac{1}{2}}$ with maximum eigenvalue, that is,
\begin{equation}
P_m\frac{C^{\frac{1}{2}}\cdot\Gamma_{id}\cdot C^{\frac{1}{2}}}{\operatorname{Tr}(\Gamma_{id}\cdot C)}P_m=\frac{C^{\frac{1}{2}}\cdot\Gamma_{id}\cdot C^{\frac{1}{2}}}{\operatorname{Tr}(\Gamma_{id}\cdot C)},
\end{equation}
more succinctly,
\begin{equation}
C^{-\frac{1}{2}}P_mC^{\frac{1}{2}}\cdot\Gamma_{id}\cdot C^{\frac{1}{2}}P_m C^{-\frac{1}{2}}=\Gamma_{id},
\end{equation}
for $C^{-\frac{1}{2}}\cdot C^{\frac{1}{2}}=C^{\frac{1}{2}}\cdot C^{-\frac{1}{2}}=\Pi$ and $\Pi\cdot\Gamma_{id}\cdot\Pi=\Gamma_{id}$.
\end{proof}

\section{Proof and Generalization of Theorem~\ref{thm2}}\label{supp_sec:C}
By the proof of Theorem~\ref{thm1} (see Appendix~\ref{supp_sec:B}), it is clear that 
\begin{equation}
\max_{\substack{\Lambda\in\text{EPO}(H)\\ \int \text{d}\psi \,\operatorname{Tr}[\Lambda(\mathcal{N}(\psi)^{\otimes n})]>0}}\frac{\int \text{d}\psi \,\operatorname{Tr}[\psi\cdot\operatorname{Tr}_{2\cdots n}\circ\Lambda(\mathcal{N}(\psi)^{\otimes n})]}{\int \text{d}\psi \,\operatorname{Tr}[\Lambda(\mathcal{N}(\psi)^{\otimes n})]}=\left\lVert C^{-\frac{1}{2}} A C^{-\frac{1}{2}} \right\rVert_\infty,
\end{equation}
which validates Theorem~\ref{thm2} immediately.

Notably, our analysis extends beyond determining the maximum average fidelity to providing a complete parameterization of all universal energy-preserving purification protocols. In the following theorem, we generalize Theorem~\ref{thm2} to characterize the Choi operators of these protocols. This parameterization explicitly presents the average fidelity and success probability of any valid protocol, enabling the construction of every protocol that delivers a specified legal average fidelity and success probability. Consequently, this formulation not only demonstrates the maximum average fidelity, but also elucidates the full landscape of achievable performances, offering a systematic method to construct every protocol that performs at any permissible level.

\begin{theorem}[Generalization of Theorem~\ref{thm2}]
\label{thm6}
Let $\cN$ satisfy Eq.~\eqref{eq:noise-condition}, and let $H$ be a given Hamiltonian. The $n\rightarrow1$ universal energy-preserving purification protocols $\mathcal{E} = \operatorname{Tr}_{2\cdots n} \circ \Lambda$ are completely characterized by the Choi operators
\begin{equation}\label{thm6-1}
\Gamma_{\Lambda} = q \, C^{-\frac{1}{2}} \, \sigma \, C^{-\frac{1}{2}},
\end{equation}
parameterized by state $\sigma$ and scalar $q$ subject to the following three constraints:
\begin{enumerate}
    \item \ $\text{Supp}(\sigma)\subseteq\text{Supp}(C)$;
    \item \ $\operatorname{Tr}(\sigma \cdot C^{-\frac{1}{2}} A C^{-\frac{1}{2}}) > \int \text{d}\psi \, \langle \psi | \mathcal{N}(\psi) | \psi \rangle$;
    \item \ $0 < q \leq \left\lVert \operatorname{Tr}_{\text{out}} \left( C^{-\frac{1}{2}} \sigma C^{-\frac{1}{2}} \right) \right\rVert_\infty^{-1}$,
\end{enumerate}
where $A$, $C$ are defined as in Theorem \ref{thm1} encoding the information of $(\mathcal{N}, H, n)$, $\operatorname{Tr}_{\text{out}}$ is the partial trace over the output system.
\\
In this representation, the performance metrics are given by $\bar{F}(\mathcal{N},\cE, H,n) = \operatorname{Tr}(\sigma \cdot C^{-\frac{1}{2}} A C^{-\frac{1}{2}})$ and $\bar{p}(\mathcal{N},\cE, H,n) = q$. Consequently, the maximum average fidelity is
\begin{equation}
F_{\text{max}}(\mathcal{N}, H, n) = \left\lVert C^{-\frac{1}{2}} A C^{-\frac{1}{2}} \right\rVert_\infty,
\end{equation}
and the set of achievable average purification fidelities constitutes an interval $$(\int \text{d}\psi \, \langle \psi | \mathcal{N}(\psi) | \psi \rangle, \, \left\lVert C^{-\frac{1}{2}} A C^{-\frac{1}{2}} \right\rVert_\infty].$$
\end{theorem}

\begin{proof}
\\
(i). We claim that any operation $\Lambda$ whose Choi operator takes the form of Eq.~\eqref{thm6-1} constitutes a valid universal energy-preserving purification protocol, and the performance metrics are $\bar{F}(\mathcal{N},\cE, H,n) = \operatorname{Tr}(\sigma \cdot C^{-\frac{1}{2}} A C^{-\frac{1}{2}})$ and $\bar{p}(\mathcal{N},\cE, H,n) = q$: First it is direct to see that $\Gamma_\Lambda$ is indeed a valid Choi operator, and we have
\begin{equation}
\Pi\cdot\Gamma_\Lambda\cdot \Pi= q \Pi\cdot\, C^{-\frac{1}{2}} \, \sigma \, C^{-\frac{1}{2}}\cdot\Pi=q \, C^{-\frac{1}{2}} \, \sigma \, C^{-\frac{1}{2}}=\Gamma_\Lambda,
\end{equation}
for $\Pi=C^{\frac{1}{2}}\cdot C^{-\frac{1}{2}}=C^{-\frac{1}{2}}\cdot C^{\frac{1}{2}}$ (see Eq.~\eqref{thm 1-8}), where $\Pi = \sum_E P_E^T \otimes P_E$ is the projector constructed from the spectral decomposition $H=\sum_E E P_E$. This justifies that $\Lambda\in\text{EPO}(H)$. Note that
\begin{equation}
 \int \text{d}\psi \, \operatorname{Tr}[\cE(\mathcal{N}(\psi)^{\otimes n})]=\int \text{d}\psi \, \operatorname{Tr}[\Lambda(\mathcal{N}(\psi)^{\otimes n})]=\operatorname{Tr}(\Gamma_\Lambda\cdot C)=q\cdot\operatorname{Tr}(\sigma\cdot\Pi)=q>0
\end{equation}
for $\text{Supp}(\sigma)\subseteq\text{Supp}(C)=\text{Supp}(\Pi)$, and
\begin{equation}
\frac{\int \text{d}\psi \,\operatorname{Tr}[\psi\cdot\operatorname{Tr}_{2\cdots n}\circ\Lambda(\mathcal{N}(\psi)^{\otimes n})]}{\int \text{d}\psi \,\operatorname{Tr}[\Lambda(\mathcal{N}(\psi)^{\otimes n})]}=\frac{\operatorname{Tr}(\Gamma_\Lambda\cdot A)}{\operatorname{Tr}(\Gamma_\Lambda\cdot C)}=\operatorname{Tr}(\sigma \cdot C^{-\frac{1}{2}} A C^{-\frac{1}{2}}) > \int \text{d}\psi \, \langle \psi | \mathcal{N}(\psi) | \psi \rangle,
\end{equation}
which justifies that $\Lambda$ constitutes a valid universal energy-preserving purification protocol and $\bar{p}(\mathcal{N},\cE, H,n) = q$, $\bar{F}(\mathcal{N},\cE, H,n) = \operatorname{Tr}(\sigma \cdot C^{-\frac{1}{2}} A C^{-\frac{1}{2}})$;\\ \\
(ii). We now show the opposite direction, that is, for any universal energy-preserving purification protocol $\mathcal{E} = \operatorname{Tr}_{2\cdots n} \circ \Lambda$, there exist state $\sigma$ and scalar $q$ with the three constraints, such that $\Gamma_{\Lambda} = q \, C^{-\frac{1}{2}} \, \sigma \, C^{-\frac{1}{2}}$, $\bar{F}(\mathcal{N},\cE, H,n) = \operatorname{Tr}(\sigma \cdot C^{-\frac{1}{2}} A C^{-\frac{1}{2}})$ and $\bar{p}(\mathcal{N},\cE, H,n) = q$: First note that the average success probability can be written as
\begin{equation}
\bar{p}(\mathcal{N},\cE, H,n) = \operatorname{Tr}(\Gamma_\Lambda\cdot C)>0,
\end{equation}
we can construct the state $\sigma$ as
\begin{equation}
\sigma\coloneqq\frac{C^{\frac{1}{2}}\cdot\Gamma_{\Lambda}\cdot C^{\frac{1}{2}}}{\operatorname{Tr}(\Gamma_{\Lambda}\cdot C)},
\end{equation}
then we have both $\text{Supp}(\sigma)\subseteq\text{Supp}(C)=\text{Supp}(\Pi)$ and
\begin{align}
\operatorname{Tr}(\sigma \cdot C^{-\frac{1}{2}} A C^{-\frac{1}{2}})&=\frac{\operatorname{Tr}(\Gamma_\Lambda\cdot A)}{\operatorname{Tr}(\Gamma_\Lambda\cdot C)}\\
&=\frac{\int \text{d}\psi \,\operatorname{Tr}[\psi\cdot\operatorname{Tr}_{2\cdots n}\circ\Lambda(\mathcal{N}(\psi)^{\otimes n})]}{\int \text{d}\psi \,\operatorname{Tr}[\Lambda(\mathcal{N}(\psi)^{\otimes n})]}\\
&=\bar{F}(\mathcal{N},\cE, H,n)\\
&>\int \text{d}\psi \, \langle \psi | \mathcal{N}(\psi) | \psi \rangle.
\end{align}
Choose $q\coloneqq\bar{p}(\mathcal{N},\cE, H,n) = \operatorname{Tr}(\Gamma_\Lambda\cdot C)>0$, then we have $\Gamma_{\Lambda} = q \, C^{-\frac{1}{2}} \, \sigma \, C^{-\frac{1}{2}}$ for $C^{\frac{1}{2}}\cdot C^{-\frac{1}{2}}=C^{-\frac{1}{2}}\cdot C^{\frac{1}{2}}=\Pi$ and $\Pi\cdot\Gamma_\Lambda\cdot\Pi=\Gamma_\Lambda$. Furthermore, since $\Gamma_\Lambda$ is a valid Choi operator, we must have
\begin{equation}
\operatorname{Tr}_{\text{out}}(\Gamma_\Lambda)=q \operatorname{Tr}_{\text{out}} \left( C^{-\frac{1}{2}} \sigma C^{-\frac{1}{2}} \right) \leq I,
\end{equation}
which implies $q\left\lVert \operatorname{Tr}_{\text{out}} \left( C^{-\frac{1}{2}} \sigma C^{-\frac{1}{2}} \right) \right\rVert_\infty\leq1$, that is, $q \leq \left\lVert \operatorname{Tr}_{\text{out}} \left( C^{-\frac{1}{2}} \sigma C^{-\frac{1}{2}} \right) \right\rVert_\infty^{-1}$ (note that we always have $\left\lVert \operatorname{Tr}_{\text{out}} \left( C^{-\frac{1}{2}} \sigma C^{-\frac{1}{2}} \right) \right\rVert_\infty>0$ for $\Gamma_\Lambda\neq0$);\\ \\
(iii). Finally, by (i) and (ii), the set of achievable average purification fidelities is precisely
\begin{equation}
\{\operatorname{Tr}(\sigma \cdot C^{-\frac{1}{2}} A C^{-\frac{1}{2}}):\sigma\geq0,\operatorname{Tr}\sigma=1,\text{Supp}(\sigma)\subseteq\text{Supp}(C),\operatorname{Tr}(\sigma \cdot C^{-\frac{1}{2}} A C^{-\frac{1}{2}}) > \int \text{d}\psi \, \langle \psi | \mathcal{N}(\psi) | \psi \rangle\},
\end{equation}
i.e., a continuous interval $(\int \text{d}\psi \, \langle \psi | \mathcal{N}(\psi) | \psi \rangle, \, \left\lVert C^{-\frac{1}{2}} A C^{-\frac{1}{2}} \right\rVert_\infty]$. Specifically, we have $F_{\text{max}}(\mathcal{N}, H, n) = \left\lVert C^{-\frac{1}{2}} A C^{-\frac{1}{2}} \right\rVert_\infty$.
\end{proof}

By the Theorem~\ref{thm6}, the optimal fidelity set $\text{UEPP}^*(\mathcal{N}, H, n)$ can be characterized as 
\begin{align}
    \text{UEPP}^*(\mathcal{N}, H, n)=\{\operatorname{Tr}_{2\cdots n} \circ \Lambda:&\Gamma_{\Lambda} = q \, C^{-\frac{1}{2}} \, \sigma \, C^{-\frac{1}{2}}, \operatorname{Tr}(\sigma \cdot C^{-\frac{1}{2}} A C^{-\frac{1}{2}})=\left\lVert C^{-\frac{1}{2}} A C^{-\frac{1}{2}} \right\rVert_\infty,\\
    &0 < q \leq \left\lVert \operatorname{Tr}_{\text{out}} \left( C^{-\frac{1}{2}} \sigma C^{-\frac{1}{2}} \right) \right\rVert_\infty^{-1}\}.
\end{align}

\section{Purification with a given energy cost}\label{supp_sec:G}

For purification with a given energy cost, a battery is attached to the $n$-copy system. Write the Hilbert space of the battery system as $\mathcal{H}_r$. Denote the state of the battery as $\varphi\in\mathcal{D}(\mathcal{H}_r)$. Let the Hamiltonian governing the composite system of the battery and the $n$ noisy copies be $\widetilde{H}\in\mathcal{B}(\mathcal{H}_r\otimes\mathcal{H}^{\otimes n})$. Similarly to the zero energy cost scenarios (see Definition~\ref{def1}), we define the $n\rightarrow 1$ \textit{universal energy-constrained purification protocols} as those operations of the form
\begin{equation}
\cE=\operatorname{Tr}_{\mathcal{H}_r}\circ\operatorname{Tr}_{2...n}\circ\widetilde{\Lambda},
\end{equation}
with energy-preserving $\widetilde{\Lambda}\in\text{EPO}(\widetilde{H})$, nonvanishing average success probability
\begin{equation}
\bar{p}(\mathcal{N},\cE, \widetilde{H},\varphi,n)\coloneqq\int \text{d}\psi \,\operatorname{Tr}[\cE(\varphi\otimes\mathcal{N}(\psi)^{\otimes n})]>0,
\end{equation}
and average fidelity improvement
\begin{equation}
\bar{F}(\mathcal{N},\cE, \widetilde{H}, \varphi,n)\coloneqq\frac{\int \text{d}\psi \,\operatorname{Tr}[\psi \cdot \cE(\varphi\otimes\mathcal{N}(\psi)^{\otimes n})]}{\bar{p}(\mathcal{N},\cE, \widetilde{H},\varphi,n)}> \int d\psi \, F(\cN(\psi), \psi).
\end{equation}
For such $(\cN, \widetilde{H}, \varphi, n)$, denote the set of all $n\rightarrow1$ universal energy-constrained purification protocols as \text{UECP}$(\cN, \widetilde{H}, \varphi, n)$. In this regard, the maximum average fidelity reads
\begin{align}
F_{\text{max}}(\mathcal{N}, \widetilde{H}, \varphi, n)&\coloneqq \max_{\cE\in\text{UECP}(\cN,\widetilde{H},\varphi,n)}\bar{F}(\mathcal{N},\cE, \widetilde{H}, \varphi,n)\\
&=  \max_{\cE\in\text{UECP}(\cN,\widetilde{H},\varphi,n)}\frac{\int \text{d}\psi \,\operatorname{Tr}[\psi \cdot \cE(\varphi\otimes\mathcal{N}(\psi)^{\otimes n})]}{\int \text{d}\psi \,\operatorname{Tr}[\cE(\varphi\otimes\mathcal{N}(\psi)^{\otimes n})]}.
\end{align}
Furthermore, denote the set of all $n\rightarrow1$ universal energy-constrained purification protocols that deliver $F_{\text{max}}(\mathcal{N}, \widetilde{H}, \varphi, n)$ as $\text{UECP}^*(\cN, \widetilde{H}, \varphi, n)$, then the maximum average success probability is
\begin{align}
p_{\text{succ}}^{\text{max}}(\mathcal{N}, \widetilde{H}, \varphi, n) &\coloneqq\max_{\cE\in\text{UECP}^*(\cN,\widetilde{H},\varphi,n)} \bar{p}(\mathcal{N},\cE, \widetilde{H},\varphi,n)\\
&=\max_{\cE\in\text{UECP}^*(\cN,\widetilde{H},\varphi,n)} \int \text{d}\psi \,\operatorname{Tr}[\cE(\varphi\otimes\mathcal{N}(\psi)^{\otimes n})],
\end{align}
In the following, we consider the case where $\cN$ satisfies the condition Eq.~\eqref{eq:noise-condition} and $\varphi$ is engineered as a full-rank state. Consider
\begin{equation}
\widetilde{C} \coloneqq \widetilde{\Pi} \cdot \left( \varphi\otimes \mathcal{N}^{\otimes n}(s_n) \right)^{\text{T}} \otimes I_{\mathcal{H}_r\otimes\mathcal{H}^{\otimes n}} \cdot \widetilde{\Pi},
\end{equation}
where $\widetilde{\Pi}\coloneqq\sum_E \widetilde{P}_E^T \otimes \widetilde{P}_E$ constructed from the spectral decomposition $\widetilde{H}=\sum_EE\widetilde{P}_E$. The condition $\varphi>0$ guarantees that
\begin{equation}
\text{Supp}(\widetilde{C})=\text{Supp}(\widetilde{\Pi}).
\end{equation}
This implies that Theorem~\ref{thm1}, Theorem~\ref{thm2}, Theorem~\ref{thm6}, Proposition~\ref{pro3} and Corollary~\ref{cor:rank-one} can all be extended to non-zero energy cost settings, just by following the same route. The extensions are formalized below omitting the proofs.   

\begin{theorem}\label{thm7}
Let $\mathcal{N}$ satisfy Eq.~\eqref{eq:noise-condition}, let the battery be initialized in a full-rank state $\varphi$, and let $\widetilde{H}$ be the Hamiltonian governing the battery and the $n$ noisy copies. Then there exists no $n\rightarrow1$ universal energy-constrained purification protocol
% for the triple $(\mathcal{N},H,n)$
if and only if
\begin{equation}
\widetilde{C}^{-\frac{1}{2}}\widetilde{P}_m\widetilde{C}^{\frac{1}{2}}\cdot\Gamma_{id}\cdot \widetilde{C}^{\frac{1}{2}}\widetilde{P}_m \widetilde{C}^{-\frac{1}{2}}=\Gamma_{id},
\end{equation}
where $\Gamma_{id}$ is the Choi operator of the identity channel $id$. Encapsulating the information of $(\mathcal{N},\widetilde{H},\varphi,n)$, the operators $\widetilde{A}$ and $\widetilde{C}$ are defined as
\begin{align}
    \widetilde{A} &= \varphi^{\text{T}}\otimes\left[\int \text{d}\psi \, (\mathcal{N}(\psi)^{\otimes n})^{\text{T}}\otimes I_{\mathcal{H}_r}\otimes\psi\right]\otimes I_{\mathcal{H}^{\otimes (n-1)}}, \\ 
\widetilde{C} &= \widetilde{\Pi} \cdot \left( \varphi\otimes \mathcal{N}^{\otimes n}(s_n) \right)^{\text{T}} \otimes I_{\mathcal{H}_r\otimes\mathcal{H}^{\otimes n}} \cdot \widetilde{\Pi},
\end{align}
with $\widetilde{\Pi} = \sum_E \widetilde{P}_E^T \otimes \widetilde{P}_E$ constructed from the spectral decomposition $\widetilde{H}=\sum_E E \widetilde{P}_E$; $s_k$ denotes the maximally mixed state on the symmetric subspace of $\mathcal{H}^{\otimes k}$; and $\widetilde{P}_m$ is the projector on the eigenspace of $\widetilde{C}^{-\frac{1}{2}}\cdot \widetilde{A}\cdot \widetilde{C}^{-\frac{1}{2}}$ with maximum eigenvalue.
\end{theorem}

\begin{theorem}\label{thm8}
Under the assumptions of Theorem~\ref{thm7}, the $n\rightarrow 1$ universal energy-constrained purification protocols $\cE=\operatorname{Tr}_{\mathcal{H}_r}\circ\operatorname{Tr}_{2...n}\circ\widetilde{\Lambda}$ are completely characterized by the Choi operators
\begin{equation}
\Gamma_{\widetilde{\Lambda}} = q \, \widetilde{C}^{-\frac{1}{2}}\cdot \, \widetilde{\sigma} \, \cdot\widetilde{C}^{-\frac{1}{2}},
\end{equation}
parameterized by state $\widetilde{\sigma}$ and scalar $q$ subject to the following three constraints:
\begin{enumerate}
    \item \ $\text{Supp}(\widetilde{\sigma})\subseteq\text{Supp}(\widetilde{C})$;
    \item \ $\operatorname{Tr}(\widetilde{\sigma} \cdot \widetilde{C}^{-\frac{1}{2}}\cdot \widetilde{A}\cdot \widetilde{C}^{-\frac{1}{2}}) > \int \text{d}\psi \, \langle \psi | \mathcal{N}(\psi) | \psi \rangle$;
    \item \ $0 < q \leq \left\lVert \operatorname{Tr}_{\text{out}} \left( \widetilde{C}^{-\frac{1}{2}}\cdot \widetilde{\sigma}\cdot \widetilde{C}^{-\frac{1}{2}} \right) \right\rVert_\infty^{-1}$,
\end{enumerate}
where $\widetilde{A}$, $\widetilde{C}$ are defined as in Theorem \ref{thm7} encoding $(\mathcal{N},\widetilde{H},\varphi,n)$, $\operatorname{Tr}_{\text{out}}$ is the partial trace over the output system.
\\
In this representation, the performance metrics are $\bar{F}(\mathcal{N},\cE, \widetilde{H}, \varphi,n) = \operatorname{Tr}(\widetilde{\sigma} \cdot \widetilde{C}^{-\frac{1}{2}}\cdot \widetilde{A}\cdot \widetilde{C}^{-\frac{1}{2}})$ and $\bar{p}(\mathcal{N},\cE, \widetilde{H},\varphi,n) = q$. Consequently, the maximum average fidelity is
\begin{equation}
F_{\text{max}}(\mathcal{N},\widetilde{H},\varphi,n) = \left\lVert \widetilde{C}^{-\frac{1}{2}}\cdot \widetilde{A}\cdot \widetilde{C}^{-\frac{1}{2}} \right\rVert_\infty,
\end{equation}
and the set of achievable average purification fidelities is an interval 
$$(\int \text{d}\psi \, \langle \psi | \mathcal{N}(\psi) | \psi \rangle, \, \left\lVert \widetilde{C}^{-\frac{1}{2}}\cdot \widetilde{A}\cdot \widetilde{C}^{-\frac{1}{2}} \right\rVert_\infty].$$
\end{theorem}

\begin{proposition}
\label{pro3_plus}
Under the assumptions of Theorem~\ref{thm7}, the maximum average success probability of $n\rightarrow 1$ universal energy-constrained purification protocols is given by the following SDP:
\begin{align}
p_{\text{succ}}^{\text{max}}(\mathcal{N}, \widetilde{H}, \varphi, n)=\textrm{maximize} \quad & \operatorname{Tr}(\Gamma_{\widetilde{\Lambda}}\cdot \widetilde{C}) \label{pro 1}\\
\textrm{subject to} \quad & \Gamma_{\widetilde{\Lambda}} \geq 0, \label{pro 2}\\
& \operatorname{Tr}_{\text{out}}(\Gamma_{\widetilde{\Lambda}}) \leq I_{\mathcal{H}_r\otimes\mathcal{H}^{\otimes n}}, \label{pro 3}\\
 \widetilde{C}^{-\frac{1}{2}} \widetilde{P}_m \widetilde{C}^{\frac{1}{2}} \cdot &\Gamma_{\widetilde{\Lambda}} \cdot \widetilde{C}^{\frac{1}{2}} \widetilde{P}_m \widetilde{C}^{-\frac{1}{2}} = \Gamma_{\widetilde{\Lambda}}, \label{pro 4}
\end{align}
where $\widetilde{C}$, $\widetilde{P}_m$, $\operatorname{Tr}_{\text{out}}$ are all defined as in Theorem \ref{thm8}. Moreover, the solution $\Gamma_{\widetilde{\Lambda}^*}$ of the SDP represents the Choi operator of $\widetilde{\Lambda}^*$, yielding the optimal protocol $\cE^*=\operatorname{Tr}_{\mathcal{H}_r}\circ\operatorname{Tr}_{2\cdots n} \circ \widetilde{\Lambda}^*$ that captures both $F_{\text{max}}(\mathcal{N}, \widetilde{H}, \varphi,n)$ and $p_{\text{succ}}^{\text{max}}(\mathcal{N}, \widetilde{H}, \varphi, n)$.
\end{proposition}

\begin{corollary}\label{Cor:rank-one}
If $\widetilde{P}_m$ has rank one for its setting $(\mathcal{N},\widetilde{H},\varphi,n)$, then the optimal universal energy-constrained purification protocol $\cE^*=\operatorname{Tr}_{\mathcal{H}_r}\circ\operatorname{Tr}_{2...n}\circ\widetilde{\Lambda}^*$ is characterized by $\widetilde{\Lambda}^*$ with Choi operator
\begin{equation}
\Gamma_{\widetilde{\Lambda}^*}=\frac{\widetilde{C}^{-1/2}\widetilde{P}_m\widetilde{C}^{-1/2}}{ \|\operatorname{Tr}_{\rm out}\widetilde{C}^{-1/2}\widetilde{P}_m\widetilde{C}^{-1/2}\|_\infty}.
\end{equation}
And the corresponding maximum average success probability is
\begin{equation}
p_{\text{succ}}^{\text{max}}(\mathcal{N},\widetilde H,\varphi,n)= \|\operatorname{Tr}_{\rm out}\widetilde{C}^{-1/2}\widetilde{P}_m\widetilde{C}^{-1/2}\|_\infty^{-1}.
\end{equation}
\end{corollary}

\end{document}